\documentclass[envcountsame]{llncs}

\usepackage{amsfonts}
\usepackage{amsmath}
\usepackage{amssymb}
\usepackage{txfonts}
\usepackage{color}
\usepackage{graphics}
\usepackage{tikz}
\usetikzlibrary{arrows,shapes,backgrounds}
\usepackage[ruled,linesnumbered]{algorithm2e}

\tikzstyle{max}=[thick,draw,minimum size=1.4em,inner sep=0em]
\tikzstyle{min}=[diamond,thick,draw,minimum size=1.4em,%
    inner sep=0em]
\tikzstyle{ran}=[circle,thick,draw,minimum size=1.4em,%
    inner sep=0em]
\tikzstyle{mc}=[rounded corners,thick,draw,minimum size=1.4em,%
    inner sep=.5ex]
\tikzstyle{tran}=[thick,draw,->,>=stealth]
\tikzstyle{loop left}=[tran, to path={.. controls +(150:.5)
    and +(210:.5) .. (\tikztotarget) \tikztonodes}]
\tikzstyle{loop right}=[tran, to path={.. controls +(30:.5)
    and +(330:.5) .. (\tikztotarget) \tikztonodes}]
\tikzstyle{loop above}=[tran, to path={.. controls +(60:.5)
    and +(120:.5) .. (\tikztotarget) \tikztonodes}]
\tikzstyle{loop below}=[tran, to path={.. controls +(240:.5)
    and +(300:.5) .. (\tikztotarget) \tikztonodes}]

\newcommand{\odrazkaaaa}{$*$}
\newcommand{\mylabel}{a}
\newcount\cvycet
\newlength{\odsazeni}
\renewenvironment{itemize}%
  {\advance\cvycet by 1
   \ifnum\cvycet=1 \renewcommand{\mylabel}{\makebox[0ex][r]{\odrazkaa}}%
                   \setlength{\odsazeni}{2.2ex}\fi%
   \ifnum\cvycet=2 \renewcommand{\mylabel}{\makebox[0ex][r]{\odrazkaaa}}%
                   \setlength{\odsazeni}{1.5ex}\fi%
   \ifnum\cvycet=3 \renewcommand{\mylabel}{\makebox[0ex][r]{\odrazkaaaa}}
                   \setlength{\odsazeni}{1.5ex}\fi%
   \begin{list}{\mylabel}{\setlength{\itemsep}{0ex}\setlength{\topsep}{.5ex}%
        \setlength{\parsep}{0pt}\setlength{\parskip}{0ex}%
        \setlength{\itemindent}{0pt}\setlength{\labelwidth}{4ex}%
        \setlength{\partopsep}{0pt}\setlength{\listparindent}{3ex}%
        \setlength{\leftmargin}{\the\odsazeni}}}%
  {\end{list}\advance\cvycet by -1}

\newcommand{\tran}[1]{{}\mathchoice%
    {\stackrel{#1}{\rightarrow}}
    {\mathop {\smash\rightarrow}\limits^{\vrule width 0pt height 0pt
                                                depth 4pt\smash{#1}}}
    {\stackrel{#1}{\rightarrow}}
    {\stackrel{#1}{\rightarrow}}
{}}
\newcommand{\ltran}[1]{{}\mathchoice%
    {\stackrel{#1}{\longrightarrow}}
    {\mathop {\smash\longrightarrow}\limits^{\vrule width 0pt height 0pt
                                                depth 4pt\smash{#1}}}
    {\stackrel{#1}{\longrightarrow}}
    {\stackrel{#1}{\longrightarrow}}
{}}

\newcommand{\A}{\mathcal{A}}

\newcommand{\calO}{\mathcal{O}}
\newcommand{\calS}{\mathcal{S}}

\newcommand{\fpath}{\mathit{FPath}}
\newcommand{\run}{\mathit{Run}}

\newcommand{\play}{\mathit{Play}}
\newcommand{\safe}{\mathit{Safe}}
\newcommand{\cover}{\mathit{Cover}}
\newcommand{\scomp}{\textit{S-comp}}
\newcommand{\ccomp}{\textit{C-comp}}

\newcommand{\len}[1]{\mathit{len}(#1)}

\newcommand{\Zset}{\mathbb{Z}}

\newcommand{\T}{\mathcal{T}}
\newcommand{\C}{\mathcal{C}}
\newcommand{\D}{\mathcal{D}}
\newcommand{\B}{\mathcal{B}}

\newcommand{\gtran}{\mapsto}
\newcommand{\coNP}{\mathbf{coNP}}
\newcommand{\NP}{\mathbf{NP}}
\newcommand{\DP}{\mathbf{DP}}
\newcommand{\PSPACE}{\mathbf{PSPACE}}
\newcommand{\EXPTIME}{\mathbf{EXPTIME}}
\newcommand{\EXPSPACE}{\mathbf{EXPSPACE}}
\newcommand{\PTIME}{\mathbf{P}}
\newcommand{\zerog}{\mathit{Zero}}
\newcommand{\firstom}{\mathit{first}}

\newcommand{\theoremlike}[2]{\par\medskip\penalty-250%
{{\bfseries\noindent
#2 \ref{#1}.}}\it}

\newcommand{\thmhelperpre}[2]{\theoremlike{#1}{#2}}
\newcommand{\thmhelperpost}{\par\medskip}

\begin{document}

\title{Efficient Controller Synthesis for Consumption Games with Multiple 
  Resource Types}
\titlerunning{Consumption Games}

\author{Tom\'{a}\v{s} Br\'{a}zdil\inst{1}$^{\star}$ \and
        Krishnendu Chatterjee\inst{2}$^{\star}$ \and
        Anton\'{\i}n Ku\v{c}era\inst{1}$^{\star}$ \and
        Petr Novotn\'{y}\inst{1}$^{\star}$}    
\authorrunning{Br\'{a}zdil, Chatterjee, Ku\v{c}era, Novotn\'{y}} 
\institute{Faculty of Informatics, Masaryk University, Brno, Czech Republic.\\
    \texttt{\{brazdil,kucera,xnovot18\}@fi.muni.cz} \and
     IST Austria, Klosterneuburg, Austria.\\
  \texttt{krish.chat@gmail.com}}

\maketitle

\begin{abstract}
\noindent\let\thefootnote\relax\footnote{\makebox[0ex][r]{$^\star~$}%
  Tom\'{a}\v{s} Br\'{a}zdil, Anton\'{\i}n Ku\v{c}era, and Petr
  Novotn\'{y} are supported by the Czech Science Foundation, grant
  No.~P202/10/1469. Krishnendu Chatterjee is supported by the FWF
  (Austrian Science Fund) NFN Grant No S11407-N23 (RiSE) and ERC Start
  grant (279307: Graph Games).}%
We introduce \emph{consumption games}, a model for discrete
interactive system with multiple resources that are consumed or reloaded
independently. More precisely, a consumption game is 
a finite-state graph where each transition is labeled by a vector
of resource updates, where every update is a non-positive number 
or~$\omega$. The $\omega$ updates model the reloading of a given 
resource. Each vertex belongs either to player~$\Box$ or 
player~$\Diamond$, where the aim of player~$\Box$ is to play so that
the resources are never exhausted. We consider several natural 
algorithmic problems about consumption games, and show that although
these problems are computationally hard in general, they are 
solvable in polynomial time for every fixed number of resource types
(i.e., the dimension of the update vectors) and bounded resource updates.  
\end{abstract}

 \pagenumbering{arabic}
\section{Introduction}
\label{sec-intro}

In this paper we introduce \emph{consumption games}, a model for discrete
interactive systems with multiple resources that can be consumed and
reloaded independently. We show that consumption games, despite their 
rich modelling power, still admit efficient algorithmic analysis
for a ``small'' number of resource types. This property distinguishes
consumption games from other related models, such as games 
over vector addition systems or multi-energy games (see below), that are 
notoriously intractable.

Roughly speaking, a consumption game is a finite-state directed graph
where each state belongs either to player~$\Box$ (controller) or
player~$\Diamond$ (environment).  Every transition $s \tran{} t$ is
labeled by a $d$-dimensional vector $\delta$ such that each component
$\delta(i)$, where $1 \leq i \leq d$, is a non-positive integer
(encoded in \emph{binary}) or~$\omega$.  Intuitively, if $\delta(i) =
-n$, then the current load of the $i$-th resource is decreased by $n$
while performing $s \tran{} t$, and if $\delta(i) = \omega$, then the
$i$-th resource can be ``reloaded'' to an arbitrarily high
value greater than or equal to the current load.  
A \emph{configuration} of a consumption game is
determined by the current control state and the current load of all
resources, which is a $d$-dimensional vector of positive integers. A
play of a consumption game is initiated in some state and some initial
load of resources.  The aim of player~$\Box$ is to play \emph{safely},
i.e., select transitions in his states so that the vector of current
resource loads stays positive in every component (i.e., the resources
are never exhausted). Player~$\Diamond$ aims at the opposite.

The resources may correspond to fuel, electricity, money, or even more abstract 
entities such as time or patience. To get a better intuition behind
consumption games and the abstract problems studied in this paper, let
us discuss one particular example in greater detail. 

The public transport company of Brno city\footnote{DPMB, Dopravn\'{\i} Podnik
M\v{e}sta Brna.} maintains the network of public trams, buses, trolleybuses,
and boats. Due to the frequent failures and breakdowns in electrical
wiring, rails, railroad switches, and the transport vehicles themselves, 
the company has several emergency teams which travel from one accident 
to another according 
to the directives received from the central supervisory service. 
Recently, the company was considering the possibility of replacing their
old diesel vans by new cars equipped with more ecological natural gas 
engines. The  problem is that these cars have smaller range and 
can be tanked only at selected gas stations. So, it is not clear
whether the cars are usable at all, i.e., whether they can always
visit a gas station on time regardless where and when an accident 
happens, and what are the time delays caused by detours to gas 
stations. Now we indicate how to construct the associated consumption 
game model and how to rephrase the above questions in abstract terms.  

We start with a standard graph $G$ representing the city road
network, i.e., the nodes of $G$ correspond to distinguished locations
(such as crossings) and the edges correspond to the connecting roads.
Then we identify the nodes corresponding to gas stations that sell
natural gas, and to each edge (road) we assign two negative numbers
corresponding to the expected time and fuel needed to pass the
road. Every morning, a car leaves a central garage (where it is fully
tanked) and returns to the same place in the evening. The maximal
number of accidents serviced per day can be safely overestimated
by~$12$. Our consumption game $\C$ has two resource types 
modeling the fuel and time in the expected way. The fuel is
consumed by passing a transition (road), and can be reloaded 
by the outgoing transitions of gas stations. The time
is also consumed by passing the roads, and the only node where it can 
be reloaded is
the central garage, but only after completing the $12$~jobs. In the
states of $\C$ we remember the current job number (from $1$
to $12$) and the current target node. At the beginning, and also after
visiting the current target node, the next target node is selected
by player~$\Diamond$. Technically, the current target node
belongs to player $\Diamond$, and there is a transition for every 
(potential) next target node. Performing such a transition does not 
consume the resources, but the information  about the next target node 
is stored in the chosen state, job index is increased,
and the control over the play is given back to player~$\Box$ who models 
the driver. This goes on until the job index reaches~$12$. Then, 
player~$\Diamond$ makes no further choice, but it is possible to reload
the time resource at the node corresponding to the central garage, and
hence  player~$\Box$ aims at returning to this place as quickly as possible
(without running out of gas). Note that $\C$ has about $12 \cdot n^2$ 
states, where $n$ is the number of states of~$G$.

The question whether the new cars are usable at all can now be formalized 
as follows: \emph{Is there is \emph{safe} strategy for player~$\Box$ in the   
initial configuration such that the fuel resource is never reloaded
to a value which is higher than the tank capacity of the car?} 
In the initial configuration, the fuel resource is initialized to~$1$ 
because it can be immediately reloaded in the central garage, and the 
time resource is initialized to a ``sufficiently high value'' which is 
efficiently computable due to  the \emph{finite reload property} 
formulated in Corollary~\ref{cor-frp}.
Similarly, the extra time delays caused by detours to gas stations can 
be estimated by computing the \emph{minimal initial credit} for 
the time resource, i.e., the minimal initial value sufficient for performing
a safe strategy, and comparing this number with the minimal initial 
credit for the time resource in a simplified consumption game where 
the fuel is not  consumed at all (this corresponds to an ideal
``infinite tank capacity''). Similarly, one could also analyze the extra
fuel costs, or model the consumption of the material needed to 
perform the repairs, and many other aspects. 

An important point of the above example is that the number 
of resources is relatively small, but the number of states is very
large. This motivates the study of \emph{parameterized complexity} of 
basic decision/optimization problems for consumption games, where the
parameters are the following:
\begin{itemize}
\item $d$, the number of resources (or \emph{dimension});
\item $\ell$, the maximal finite $|\delta(i)|$ such that 
   $1 \leq i \leq d$ and $\delta$ is a label of some transition.
\end{itemize}
\smallskip 

\noindent
\textbf{Main Results.} For every state $s$ of a consumption game $\C$, 
we consider the following sets of vectors 
(see Section~\ref{sec-prelim} for precise definitions):
\begin{itemize}
\item $\safe(s)$ consists of all vectors $\alpha$ of positive integers
  such that player~$\Box$ has a safe strategy in the configuration
  $(s,\alpha)$. That is, $\safe(s)$ consists of all vectors describing
  a sufficient \emph{initial} load of all resources needed to perform 
  a safe strategy.
\item $\cover(s)$ consists of all vectors $\alpha$ of positive
  integers such that player~$\Box$ has a safe strategy $\sigma$ in the
  configuration $(s,\alpha)$ such that for every strategy $\pi$ for
  player~$\Diamond$ and every configuration $(t,\beta)$ visited during
  the play determined by $\sigma$ and $\pi$ we have that $\beta \leq
  \alpha$.  Note that 
  physical resources (such as fuel, water, electricity, etc.) are
  stored in devices with finite capacity (tanks, batteries, etc.), and
  hence it is important to know what capacities of these devices
  are sufficient for performing a safe strategy. These sufficient
  capacities correspond to the vectors of $\cover(s)$.
\end{itemize}
Clearly, both $\safe(s)$ and  $\cover(s)$ are upwards closed 
with respect to component-wise ordering. Hence, these sets 
are fully determined by their \emph{finite} sets 
of minimal elements. In this paper we aim at answering the very
basic algorithmic problems about $\safe(s)$ and $\cover(s)$, which
are the following:
\begin{enumerate}
\item[(A)] \emph{Emptiness}. For a given state $s$, decide whether 
  $\safe(s) = \emptyset$
  (or $\cover(s) = \emptyset$).
\item[(B)] \emph{Membership}. For a given state $s$ and a vector $\alpha$,
  decide whether $\alpha \in \safe(s)$ (or \mbox{$\alpha \in \cover(s)$}).
  Further, decide whether $\alpha$ is a \emph{minimal} vector 
  of $\safe(s)$ (or $\cover(s)$).
\item[(C)] \emph{Compute the set of minimal vectors} of $\safe(s)$
  (or $\cover(s)$). 
\end{enumerate}
Note that these problems subsume the questions of our motivating example.
We show that \emph{all of these problems are computationally hard}, but 
solvable in \emph{polynomial time} for every fixed choice of the
parameters~$d$ and~$\ell$ introduced above. Since the degree of
the bounding polynomial increases with the size of the parameters, we do
\emph{not} provide fixed-parameter tractability results in the usual
sense of parameterized complexity (as it is mentioned in 
Section~\ref{sec-general}, this would  imply a solution to a long-standing 
open problem in algorithmic study of graph games). Still, these results clearly show that
for ``small'' parameter values, the above problems \emph{are} practically
solvable even if the underlying graph of~$\C$ is very large. 
More precisely, we show the following for game graphs with $n$ states:
\begin{itemize}
\item The emptiness problems for $\safe(s)$ and $\cover(s)$ are
  $\coNP$-complete, and solvable in $\calO(d!\cdot n^{d+1})$ time.
\item The membership problems for $\safe(s)$ and $\cover(s)$ are
  $\PSPACE$-hard and solvable in time 
  \mbox{$|\alpha| \cdot (d \cdot \ell  \cdot n)^{\calO(d)}$}
  and \mbox{$\calO(\Lambda^2 \cdot n^2)$}, respectively,
  where  $|\alpha|$ is the encoding size of $\alpha$ and
  $\Lambda = \Pi_{i=1}^d \alpha(i)$.
\item The set of minimal elements of $\safe(s)$ and $\cover(s)$
  is computable in time 
  \mbox{$(d \cdot \ell  \cdot n)^{\calO(d)}$}
  and $(d \cdot \ell \cdot n)^{\calO(d\cdot d!)}$, respectively.
\end{itemize}

Then, in Section~\ref{sec-restricted}, we show that the complexity of some of 
the above problems can be substantially improved for two natural 
subclasses of \emph{one-player} and \emph{decreasing} consumption 
games by employing special methods. A consumption game is
\emph{one-player} if all states are controlled by player~$\Box$,
and \emph{decreasing} if every resource is either reloaded or decreased 
along every cycle in the graph of $\C$. For example, the game constructed 
in our motivating example is decreasing, and we give a motivating
example for one-player consumption games in Section~\ref{sec-restricted}.
In particular, we prove that
\begin{itemize}
\item the emptiness problem for $\safe(s)$ and $\cover(s)$ is solvable
  in polynomial time both for one-player and decreasing consumption 
  games;
\item the membership problem for $\safe(s)$ is $\PSPACE$-complete 
  (resp. $\NP$-complete) for decreasing consumption games (resp. one-player 
  consumption games).
\item Furthermore, for both these subclasses we present algorithms to 
  compute the minimal elements of $\safe(s)$ by a reduction to 
  \emph{minimum multi-distance reachability} problem, and solving the minimum 
  multi-distance reachability problem on game graphs.
  Though the algorithms do not improve the worst case complexity 
  over general consumption games, the algorithms are iterative and 
  potentially terminate much earlier (we refer to Section~\ref{subsec-oneplayer}
  and Section~\ref{subsec-decreasing} for details).
\end{itemize}
\smallskip

\noindent
\textbf{Related Work.}  Our model of consumption games is related but
incomparable to energy games studied in the literature. In energy
games both positive and non-positive weights are allowed, but in
contrast to consumption games there are no $\omega$-weights.
Energy games with single resource type were introduced
in~\cite{CdAHS03}, and it was shown that the minimal initial credit
problem (and also the membership problem) can be solved in exponential
time.  It follows from the results of~\cite{CdAHS03} that the related
problem of existence of finite initial credit (which is same as the
emptiness problem), which was also shown to be equivalent to
two-player mean-payoff games~\cite{BFLM10}, lies in $\NP \cap \coNP$.
Games over extended vector addition systems with states (eVASS games), 
where the weights in transition labels are in $\{-1,0,1,\omega\}$, 
were introduced and studied in \cite{BJK}. In this paper, it 
was shown that the question whether player~$\Box$ has a safe strategy
in a given configuration is decidable, and the winning region of 
player~$\Box$ is computable in $d$-$\EXPTIME$, where $d$ is the 
eVASS dimension, and hence the provided solution is impractical even for
very small $d$'s. A closely related model of energy games with 
multiple resource types (or multi-energy games) were
considered in~\cite{CDHR10}.  The minimal initial credit problem (and
also the membership problem) for multi-energy games can be reduced to
the problem of games over eVASS with an exponential reduction to encode
the integer weights into weights $\{-1,0,1\}$.  Thus the minimal
initial credit problem can be solved in $d$-$\EXPTIME$, and the
membership problem is $\EXPSPACE$-hard (the hardness follows from the
classical result of Lipton~\cite{Lipton76}).  The existence of finite
initial credit problem (i.e., the emptiness problem) is
$\coNP$-complete for multi-energy games~\cite{CDHR10}.  
Thus the complexity of the membership and the minimal initial credit problem 
for consumption games is much better (it is in $\EXPTIME$ and $\PSPACE$-hard and can be solved in
polynomial time for every fixed choice of the parameters) 
as compared to eVASS games or multi-energy games ($\EXPSPACE$-hard and can be 
solved in $d$-$\EXPTIME$).  
For eVASS games with fixed dimensions, the
problem can be solved in polynomial time for $d=2$ (see \cite{Chaloupka}),
and it is open whether the complexity can be improved for other constants.
Moreover, for the important subclasses of one-player and decreasing 
consumption games we show much better bounds (polynomial time 
algorithms for emptiness and optimal complexity bounds for membership in 
$\safe(s)$).
The complexity bounds are summarized in Table~\ref{tb:complexity}, along with known bounds for energy games.

\begin{table}%
\label{tb:complexity}
\centering
\caption{Complexities of studied problems (A)-(C). Empty cell indicates problem that was not, to our best knowledge, considered for energy games.}
\begin{tabular}{|l||p{2.75cm}|p{2cm}|p{1.95cm}|p{2.15cm}|}
 \hline
 Problem & Energy games & Consumption games &  {One-player CGs} & {Decreasing CGs} \\
 \hline 
 \hline
 \parbox[t]{2.6cm}{Emptiness of $\safe(s)$ and $\cover(s)$} & \mbox{$\coNP$-compl.} \cite{CDHR10} & \mbox{$\coNP$-compl.}& $\in\PTIME$ & $\in\PTIME$\\
 \hline \hline 
 \parbox[t]{2.6cm}{Membership in\\ $\safe(s)$ }& \mbox{$\EXPSPACE$-hard}, \mbox{$\in d$-$\EXPTIME$} \cite{BJK,CDHR10} & \mbox{$\PSPACE$-hard}, \mbox{$\in\EXPTIME$} & \mbox{$\NP$-compl.} & \mbox{$\PSPACE$-compl.} \\
 \hline
 \parbox[t]{2.6cm}{Membership in \\set of minimums\\ of $\safe(s)$} & \mbox{$\EXPSPACE$-hard}, \mbox{$\in d$-$\EXPTIME$} \cite{BJK,CDHR10} & \mbox{$\PSPACE$-hard}, \mbox{$\in\EXPTIME$} & \mbox{$\DP$-compl.} & \mbox{$\PSPACE$-compl.} \\
 \hline
 \parbox[t]{2.6cm}{Membership in \\(minimums of) \\$\cover(s)$} & \mbox{$\EXPTIME$-compl.} \cite{FLLS11:EnGames} (lower bound holds even for fixed $d>1$)& \mbox{$\PSPACE$-hard}, \mbox{$\in\EXPTIME$} & \mbox{$\NP$-hard}, \mbox{$\in\EXPTIME$} & \mbox{$\PSPACE$-hard},  \mbox{$\in\EXPTIME$} \\
 \hline \hline 
 \parbox[t]{2.6cm}{Compute min. vectors in $\safe(s)$} & \mbox{$\in d$-$\EXPTIME$} \cite{BJK,CDHR10} & $(d \cdot \ell \cdot |S|)^{\calO(d)}$ & \parbox[t]{1.9cm}{$(d \cdot \ell \cdot |S|)^{\calO(d)}$,\\ iterative algorithm} &\parbox[t]{2.1cm}{$(d \cdot \ell \cdot |S|)^{\calO(d)}$,\\ iterative algorithm} \\
 \hline \parbox[t]{2.6cm}{Compute min. vectors in $\cover(s)$} & & $(d \cdot \ell \cdot |S|)^{\calO(d\cdot d!)}$ &  $(d \cdot \ell \cdot |S|)^{\calO(d\cdot d!)}$&  $(d \cdot \ell \cdot |S|)^{\calO(d\cdot d!)}$\\
 \hline
 \hline
 \multicolumn{5}{|p{11.3cm}|}{Moreover, all problems considered are tractable for every fixed choice of $d$ and $\ell$ in general consumption games. (In energy games, none of the problems is known to be tractable for $d\geq 3$.)}\\
 \hline
\end{tabular}
\end{table}

The paper is organized as follows. After presenting necessary definitions
in Section~\ref{sec-prelim}, we present our solution to the three
algorithmic problems~(A)-(C) for general consumption games
in Section~\ref{sec-general}. In Section~\ref{sec-restricted}, 
we concentrate on the two 
subclasses of decreasing and one-player consumption games and give
optimized solutions to some of these problems. Finally, in 
Section~\ref{sec-concl} we give a short list of open problems which, 
in our opinion, address some of the fundamental properties of 
consumption games that deserve further attention.
Due to the lack of space, the proofs were shifted to~Appendix.

\section{Definitions}
\label{sec-prelim}

In this paper, the set of all integers is denoted by $\Zset$.
For a given operator ${\Join} \in \{{>},{<},{\leq},{\geq}\}$, we
use $\Zset_{\Join 0}$ to denote the set $\{i \in \Zset \mid i \Join 0\}$,
and $\Zset_{\Join 0}^\omega$ to denote the set 
$\Zset_{\Join 0} \cup \{\omega\}$, where $\omega \not\in \Zset$ is a special
symbol representing an ``infinite amount'' with the usual conventions
(in particular, $c+\omega = \omega +c = \omega$ and $c < \omega$ for 
every $c \in \Zset$).
For example, $\Zset_{<0}$ is the set of all negative integers, and
$\Zset_{<0}^\omega$ is the set $\Zset_{<0} \cup \{\omega\}$. 
We use Greek letters  $\alpha,\beta,\ldots$ to denote
vectors over $\Zset_{\Join 0}$ or $\Zset_{\Join 0}^\omega$, and $\vec{0}$
to denote the vector of zeros. The
$i$-th component of a given $\alpha$ is denoted by $\alpha(i)$.
The standard component-wise ordering over vectors is denoted by $\leq$,
and we also write $\alpha < \beta$ to indicate that 
$\alpha(i) < \beta(i)$ for every~$i$.

Let $M$ be a finite or countably infinite alphabet. A \emph{word}
over~$M$ is a finite or infinite sequence of elements of~$M$. The
empty word is denoted by~$\varepsilon$, and the set of all finite
words over $M$ is denoted by $M^*$. Sometimes we also use $M^+$ to
denote the set $M^* \smallsetminus \{\varepsilon\}$.  The length of a
given word $w$ is denoted by $\len{w}$, where $\len{\varepsilon} = 0$
and the length of an infinite word is~$\infty$.  The individual
letters in a word $w$ are denoted by $w(0),w(1),\ldots$, and for every
infinite word $w$ and every $i \geq 0$ we use $w_i$ to denote the
infinite word $w(i),w(i{+}1),\ldots$.

A \emph{transition system}\index{transition system} 
is a pair $\T = (V,{\tran{}})$, where 
$V$ is a finite or countably infinite set of \emph{vertices} and 
${\tran{}} \subseteq V \times V$ a \emph{transition relation}
such that for every $v \in V$ there is at least one outgoing transition
(i.e., a transition of the form $v \tran{} u$).
A \emph{path} in $\T$ is a finite or infinite word $w$ over~$V$
such that $w(i) \tran{} w(i{+}1)$ for every $0 \leq i < \len{w}$.
We call a finite path a \emph{history} and infinite path a \emph{run}. 
The sets of all finite paths and all runs in~$\T$ 
are denoted by $\fpath(\T)$ and $\run(\T)$, respectively.

\begin{definition}
\label{def-game}
A \emph{(2-player) game} is a triple $G =
(V,\gtran,(V_{\Box},V_\Diamond))$ where $(V,{\gtran})$ is a
transition system and $(V_{\Box},V_\Diamond)$ is a partition
of~$V$. If $V_\Diamond = \emptyset$, then $G$ is a \mbox{\emph{1-player game}}.
\end{definition}
A game $G$ is played by two players, $\Box$ and $\Diamond$,
who select transitions in the vertices of $V_{\Box}$ and $V_{\Diamond}$,
respectively. Let $\odot \in \{\Box,\Diamond\}$.
A \emph{strategy} for player~$\odot$ is a function which to each 
$wv \in V^*V_{\odot}$ assigns a state $v' \in V$ such that 
$v \gtran v'$. 
The sets of all strategies for player~$\Box$ and 
player~$\Diamond$ are denoted by 
$\Sigma_G$ and $\Pi_G$ (or just by $\Sigma$ and $\Pi$ if $G$ is understood),
respectively. We say that a strategy $\tau$
is \emph{memoryless} if $\tau(wv)$ depends just on the last
state~$v$, for every $w\in V^*$. Strategies that are not necessarily memoryless are called 
\emph{history-dependent}. 
Note that every initial vertex $v$ and every pair of strategies 
$(\sigma,\pi) \in \Sigma \times \Pi$ determine a unique 
infinite path in~$G$ initiated in~$v$, which is called a \emph{play}
and denoted by $\play_{\sigma,\pi}(v)$.

\begin{definition}
\label{def-cgame}
  Let $d \geq 1$.
  A \emph{consumption game of dimension $d$} is a tuple 
  $\C = (S,E,(S_{\Box},S_{\Diamond}),L)$ where
  $S$ is a finite set of \emph{states}, $(S,E)$ is a transition
  system, $(S_{\Box},S_{\Diamond})$ is a partition of~$S$, and 
  $L$ is \emph{labelling} which to every $(s,t) \in E$ assigns
  a vector $\delta = (\delta(1),\ldots,\delta(d))$ such that
  $\delta(i) \in \Zset_{\leq 0 }^\omega$ for every $1 \leq i \leq d$.
  If $s \in S_{\Diamond}$, we require that $\delta(i) \neq \omega$ for
  all $1 \leq i \leq d$.
  We  write $s \tran{\delta} t$ to indicate that $(s,t) \in E$ and 
  $L(s,t) = \delta$. 

  We say that $\C$ is \emph{one-player} if $S_\Diamond = \emptyset$,
  and \emph{decreasing} if for every $n \geq 1$, every  $1 \leq i \leq d$,
  and every  path
  $s_0 \ltran{\delta_1} s_1 \ltran{\delta_2} \cdots \ltran{\delta_n} s_n$
  such that $s_0 = s_n$,
  there is some $j\leq n$ where $\delta_j(i) \neq 0$.
\end{definition}

\noindent
Intuitively, if
$s \tran{\delta} t$, then the system modeled by $\C$ can move from
the state $s$ to the state $t$ so that its resources are consumed/reloaded
according to $\delta$. More precisely, if $\delta(i) \leq 0$, then 
the current load of resource~$i$ is decreased by $|\delta(i)|$, and
if $\delta(i) = \omega$, then the resource~$i$ can be reloaded to an arbitrarily 
high positive value larger than or equal to the current load. The aim of 
player~$\Box$ is to play so that
the resources are never exhausted, i.e., the vector of current loads
stays positive in every component. The aim of player~$\Diamond$ is
to achieve the opposite. 

The above intuition is formally captured by defining the associated  
infinite-state game $G_{\C}$ for $\C$. The vertices of
$G_{\C}$ are \emph{configurations} of $\C$, i.e., the elements of 
$S \times \Zset_{> 0}^d$ together with a special
configuration $F$ (which stands for ``fail''). The transition relation 
$\gtran$ of $G_{\C}$ is determined as follows:
\begin{itemize}
\item $F \gtran F$.
\item For every configuration $(s,\alpha)$ and every transition 
  $s \tran{\delta} t$ of $\C$ such that $\alpha(i) + \delta(i) > 0$
  for all $1 \leq i \leq d$, there
  is a transition $(s,\alpha) \gtran (t,\alpha{+}\gamma)$ for every
  $\gamma \in \Zset^d$ such that 
  \begin{itemize}
  \item $\gamma(i) = \delta(i)$ for every $1 \leq i \leq d$ where 
     $\delta(i) \neq \omega$;
  \item $\gamma(i) \geq 0$ for every $1 \leq i \leq d$ where 
     $\delta(i) = \omega$.
  \end{itemize}
\item If $(s,\alpha)$ is a configuration and
  $s \tran{\delta} t$ a transition of $\C$ such that 
  $\alpha(i) + \delta(i) \leq 0$
  for some $1 \leq i \leq d$, then there is a transition 
  $(s,\alpha) \gtran F$.
\item There are no other transitions.  
\end{itemize}

\noindent
A strategy $\sigma$ for player~$\Box$ in $G_{\C}$ is \emph{safe} in
a configuration $(s,\alpha)$ iff for every strategy $\pi$ for 
player~$\Diamond$ we have that $\play_{\sigma,\pi}(s,\alpha)$ does
\emph{not} visit the configuration~$F$. For every $s \in S$, we use
\begin{itemize}
\item $\safe(s)$ to denote the set of all $\alpha \in \Zset_{>0}^d$ 
   such that player~$\Box$ has a safe strategy in $(s,\alpha)$;
 \item $\cover(s)$ to denote the set of all $\alpha \in \Zset_{>0}^d$ 
   such that player~$\Box$ has a safe strategy $\sigma$ in
   $(s,\alpha)$ such that for every strategy $\pi$ for
   player~$\Diamond$ and every configuration $(t,\beta)$ visited by
   $\play_{\sigma,\pi}(s,\alpha)$ we have that $\beta \leq \alpha$.
\end{itemize}
If $\alpha \in \safe(s)$, we say that $\alpha$ is \emph{safe in $s$},
and if $\alpha \in \cover(s)$, we say that $\alpha$ is \emph{covers~$s$}.  
Obviously, $\cover(s) \subseteq \safe(s)$, and both
$\safe(s)$ and $\cover(s)$ are upwards closed w.r.t.\ component-wise
ordering (i.e., if $\alpha \in \safe(s)$ and $\alpha \leq \alpha'$,
then $\alpha' \in \safe(s)$). This means that $\safe(s)$ and
$\cover(s)$ are fully described by its finitely many minimal elements.

Intuitively, $\safe(s)$ consists of all vectors describing a sufficiently
large \emph{initial} amount of all resources needed to perform a safe
strategy. Note that during a play, the resources can be reloaded to 
values that are larger than the initial one. Since physical resources 
are stored in ``tanks'' with finite capacity, we need to know
what \emph{capacities} of these tanks are sufficient for performing
a safe strategy. These sufficient capacities are encoded by the vectors
of $\cover(s)$.

\section{Algorithms for General Consumption Games}
\label{sec-general}

In this section we present a general solution for the three algorithmic 
problems~(A)-(C) given in Section~\ref{sec-intro}. 

We start by a simple observation that connects the study of consumption 
games to a more mature theory of Streett games. A Streett
game is a tuple \mbox{$\calS = (V,\gtran,(V_{\Box},V_\Diamond),\A)$},
where $(V,\gtran,(V_{\Box},V_\Diamond))$ is a 2-player game with finitely many
vertices, and $\A = \{(G_1,R_1),\ldots,(G_m,R_m)\}$, where $m \geq 1$
and $G_i,R_i \subseteq {\gtran}$ for all $1 \leq i \leq m$, is a Streett 
(or strong fairness) winning condition (for technical convenience, we consider 
$G_i,R_i$ as subsets of edges rather than vertices). For an infinite path $w$ 
in $\calS$, let $\inf(w)$ be the
set of all edges that are executed infinitely often along~$w$.
We say that $w$ \emph{satisfies $\A$} iff $\inf(w) \cap G_i \neq \emptyset$
implies $\inf(w) \cap R_i \neq \emptyset$ for every $1 \leq i \leq m$.   
A strategy $\sigma \in \Sigma_{\calS}$ is \emph{winning} in $v \in V$ 
if for every $\pi \in \Pi_{\calS}$ we have that $\play_{\sigma,\pi}(v)$ 
satisfies $\A$. The problem whether player~$\Box$ has a winning
strategy in a vertex $v \in V$ is $\coNP$-complete~\cite{EJ88},
and the problem can be solved in 
$\calO(m! \cdot |V|^{m+1})$ time~\cite{PP06}.

For the rest of this section, we fix a consumption game 
$\C = (S,E,(S_{\Box},S_{\Diamond}),L)$ of dimension~$d$, and we use
$\ell$ to denote the maximal finite $|\delta(i)|$ such that 
$1 \leq i \leq d$ and $\delta$ is a label of some transition.
A proof of the next lemma is given in Appendix~\ref{app-Street}.

\begin{lemma}
\label{lem-con-Street}
Let $\calS_{\C} = (S,E,(S_{\Box},S_{\Diamond}),\A)$ be a Streett game where
$\A = \{(G_1,R_1),\ldots,(G_d,R_d)\}$, 
$G_i = \{(s,t) \in E \mid L(s,t)(i) < 0\}$, and
$R_i = $ \mbox{$\{(s,t) \in E \mid L(s,t)(i) = \omega\}$} for every 
$1 \leq i \leq d$. 
Then for every $s \in S$ the following assertions hold:
\begin{enumerate}
 \item If $\safe(s) \neq \emptyset$, then player~$\Box$ has a winning
      strategy in $s$ in the Streett game $\calS_{\C}$.
 \item If player~$\Box$ has a winning strategy in $s$ in the Streett game 
      $\calS_{\C}$, then 
      \mbox{$(d! \cdot |S|\cdot \ell +1 ,\ldots,d! \cdot |S|\cdot \ell +1) 
\in \safe(s) \cap \cover(s)$}.
 \end{enumerate}
\end{lemma}

\noindent
An immediate consequence of Lemma~\ref{lem-con-Street} is that 
$\safe(s) = \emptyset$ iff $\cover(s) = \emptyset$.
Our next lemma shows that the existence of a wining strategy in Streett
games is polynomially reducible to the problem whether 
$\safe(s) = \emptyset$ in consumption games.

\begin{lemma}
\label{lem-Street-con}
Let $\calS = (V,\gtran,(V_{\Box},V_\Diamond),\A)$ be a Streett game where
$\A = \{(G_1,R_1),\ldots,(G_m,R_m)\}$. Let 
$\C_{\calS} = (V,\gtran,(V_{\Box},V_\Diamond),L)$  be a consumption game
of dimension~$m$ where $L(u,v)(i)$ is either $-1$, $\omega$, or $0$,
depending on whether $(u,v) \in G_i$, $(u,v) \in R_i$, or 
$(u,v) \not\in G_i \cup R_i$, respectively. Then for every $v \in V$
we have that player~$\Box$ has a winning strategy in $v$ (in $\calS$)
iff $\safe(v) \neq \emptyset$ (in $\C_{\calS}$). 
\end{lemma}

\noindent
A direct consequence of Lemma~\ref{lem-con-Street} and 
Lemma~\ref{lem-Street-con} is the following:

\begin{theorem}
\label{thm-safe-empty}
The emptiness problems for $\safe(s)$ and $\cover(s)$ are $\coNP$-complete
and solvable in $\calO(d! \cdot |S|^{d+1})$ time.
\end{theorem}

\noindent
Also observe that if managed to prove that the emptiness problem for 
$\safe(s)$ or $\cover(s)$ is fixed-parameter tractable in~$d$ for consumption 
games where $\ell$ is equal to one (i.e.,
if we proved that the problem is solvable in time 
$F(d) \cdot n^{\calO(1)}$ where $n$ is the size of the game and $F$
a computable function), then
due to Lemma~\ref{lem-Street-con} we would immediately obtain that
the problem whether player~$\Box$ has a winning strategy in a given
Streett game is also fixed-parameter tractable. That is, we would
obtain a solution to one of the long-standing open problems of
algorithmic study of graph games.

Now we show how to compute the set of minimal elements of $\safe(s)$.
A key observation is the following lemma whose proof is
non-trivial.

\begin{lemma}
\label{lem-safe-min-bounded}
  For every $s \in S$ and every minimal $\alpha \in \safe(s)$ we have
  that $\alpha(i) \leq d \cdot \ell \cdot |S|$ for every $1 \leq i \leq d$.
\end{lemma}

\noindent
Observe that Lemma~\ref{lem-safe-min-bounded} does \emph{not} follow
from Lemma~\ref{lem-con-Street}~(2.). Apart from Lemma \ref{lem-safe-min-bounded} providing better bounds, 
Lemma~\ref{lem-con-Street}~(2.) only says that if \emph{all} 
resources are loaded enough, then there
is a safe strategy. However, we aim at proving a substantially stronger
result saying that \emph{no} resource needs to be reloaded to more
than $d \cdot \ell \cdot |S|$ \emph{regardless} how large is the 
current load of other resources. 

Intuitively, Lemma~\ref{lem-safe-min-bounded} is obtained by 
a somewhat tricky inductive argument where we first consider 
all resources as being ``sufficiently large'' and then bound the
components one by one. Since a similar technique is also used to 
compute the minimal elements of $\cover(s)$, we briefly introduce
the main underlying notions and ideas. 

An \emph{abstract load vector} $\mu$ is an element of $(\Zset_{>0}^\omega)^d$. 
The \emph{precision} of $\mu$ is the number of components 
different from~$\omega$. The standard componentwise ordering
is extended also to abstract load vectors by stipulating that $c < \omega$
for every $c \in \Zset$. Given an abstract load vector~$\mu$
and a vector $\alpha \in (\Zset_{>0})^d$, we say that $\alpha$ 
\emph{matches $\mu$} if $\alpha(j) = \mu(j)$ for all 
$1 \leq j \leq d$ such that $\mu(j) \neq \omega$. Finally, we say that 
$\mu$ is \emph{compatible} with $\safe(s)$ (or $\cover(s)$) if 
there is some $\alpha \in \safe(s)$ (or $\alpha \in \cover(s)$) 
that matches~$\mu$. 

The proof of  Lemma~\ref{lem-safe-min-bounded} is obtained by showing
that for every \emph{minimal} abstract load vector $\mu$ with precision~$i$ 
compatible with  $\safe(s)$ we have that 
$\mu(j) \leq i \cdot \ell \cdot |S|$ for every $1 \leq j \leq d$ such that
$\mu(j) \neq \omega$.  Since the minimal elements of $\safe(s)$ are 
exactly the minimal abstract vectors of precision~$d$ compatible 
with $\safe(s)$, we obtain the desired result. The claim is proven by
induction on~$i$. In the induction step, we pick a minimal abstract
vector $\mu$ with precision~$i$ compatible with~$s$, and 
choose a component~$j$ such that $\mu(j) = \omega$. Then we show
that if we replace $\mu(j)$ with some $k$ whose value is bounded 
by $(i+1) \cdot \ell \cdot |S|$, we yield
a minimal compatible abstract vector with precision~$i+1$.
The proof of this claim is the very core of the whole argument, and it
involves several
subtle observations about the structure of minimal abstract load 
vectors. The details are given in Appendix~\ref{app-abstract-load}.

An important consequence of
Lemma~\ref{lem-safe-min-bounded} is the following:

\begin{corollary}[Finite reload property]
\label{cor-frp}
  If $\alpha \in \safe(s)$ and $\beta(i) = 
  \min\{\alpha(i), d \cdot \ell \cdot |S|\}$ for every $1 \leq i \leq d$,
  then $\beta \in \safe(s)$.
\end{corollary}

\noindent
Due to Corollary~\ref{cor-frp}, for every minimal  $\alpha \in \safe(s)$
there is a safe strategy which never reloads any resource to more than 
$d \cdot \ell \cdot |S|$. Thus, we can significantly improve the bound of Lemma~\ref{lem-con-Street}~(2.).

\begin{corollary}
\label{cor-global-bound}
If $\safe(s)\neq \emptyset$, then \mbox{$(d \cdot \ell \cdot |S|, \ldots, d \cdot \ell \cdot |S|) \in  \safe(s) \cap \cover(s)$}.
\end{corollary}

\noindent
Another consequence of Corollary~\ref{cor-frp} is that
one can reduce the problem of computing
the minimal elements of $\safe(s)$ to the problem of determining a
winning set in a finite-state 2-player safety game with at most
$|S| \cdot d^d \cdot \ell^{d}  \cdot |S|^d + 1$ vertices, which is obtained
from $\C$ by storing the vector of current resource loads explicitly
in the states. Whenever we need to reload some resource, it can be
safely reloaded to $d \cdot \ell  \cdot |S|$, and we simulate this
reload be the corresponding transition. Since the winning set in a 
safety game with $n$ states and $m$ edges can be computed in time linear 
in~$n+m$ \cite{Immerman81,Beeri80}, we obtain the following:

\begin{corollary}
\label{thm-safe-compute}
  The sets of all minimal elements of all $\safe(s)$ are computable
  in time \mbox{$(d \cdot \ell  \cdot |S|)^{\calO(d)}$}.
\end{corollary}

\noindent
The complexity bounds for the algorithmic problems~(B) and~(C) 
for~$\safe(s)$ are given in our next theorem. The proofs of
the presented lower bounds are given in Appendix~\ref{app-lb}.

\begin{theorem}
\label{thm-safe}
  Let $\alpha \in \Zset_{>0}^d$ and $s \in S$.
  \begin{itemize}
  \item The problem whether $\alpha \in \safe(s)$ 
     is $\PSPACE$-hard and solvable in time 
     \mbox{$|\alpha| \cdot (d \cdot \ell  \cdot |S|)^{\calO(d)}$}, 
     where $|\alpha|$ is the encoding size of $\alpha$.
  \item The problem whether $\alpha$ is a minimal vector of $\safe(s)$ 
     is $\PSPACE$-hard and solvable in time
     \mbox{$|\alpha| \cdot (d \cdot \ell  \cdot |S|)^{\calO(d)}$}, 
     where $|\alpha|$ is the encoding size of $\alpha$.
   \item The set of all minimal vectors of $\safe(s)$
     is computable in time
     \mbox{$(d \cdot \ell  \cdot |S|)^{\calO(d)}$}
  \end{itemize}
\end{theorem}

\noindent
Now we provide analogous results for~$\cover(s)$. Note that deciding
the membership to~$\cover(s)$ is trivially reducible to the problem
of computing the winning region in a finite-sate game obtained from
$\C$ by constraining the vectors of current resource loads by $\alpha$.
Computing the minimal elements of~$\cover(s)$ is more problematic.
One is tempted to conclude that all components of the minimal vectors
for each~$\cover(s)$ are bounded by a ``small'' number, analogously
to Lemma~\ref{lem-safe-min-bounded}. In this case, we obtained only
the following bound, which is still polynomial for every fixed~$d$ and $\ell$, but
grows \emph{double-exponentially} in~$d$. The question whether this bound
can be lowered is left open, and seems to require a deeper insight
into the structure of covering vectors.

\begin{lemma}
  \label{lem-cover-min-bounded}
  For every $s \in S$ and every minimal $\alpha \in \cover(s)$ we have
  that $\alpha(i) \leq (d \cdot \ell \cdot |S|)^{d!}$ for every 
  $1 \leq i \leq d$.
\end{lemma}

\noindent 
The proof of Lemma~\ref{lem-cover-min-bounded} is given in
Appendix~\ref{app-abstract-load}. It is based
on re-using and modifying some ideas introduced in \cite{BJK}
for general eVASS games. The following theorem sums up the complexity bounds for problems~(B) and~(C) for $\cover(s)$.

\begin{samepage}
\begin{theorem}
\label{thm-cover}
  Let $\alpha \in \Zset_{>0}^d$ and $s \in S$.
  \begin{itemize}
  \item The problem whether $\alpha \in \cover(s)$ 
     is $\PSPACE$-hard and solvable in 
     \mbox{$\calO(\Lambda^2 \cdot |S|^2)$} time, where 
     $\Lambda = \Pi_{i=1}^d \alpha(i)$.
  \item The problem whether $\alpha$ is a minimal element of $\cover(s)$ 
     is $\PSPACE$-hard and solvable in 
     \mbox{$\calO(d \cdot \Lambda^2 \cdot |S|^2)$} time, where 
     $\Lambda = \Pi_{i=1}^d \alpha(i)$.
  \item The set of all minimal vectors of $\cover(s)$ is 
     computable in $(d \cdot \ell \cdot |S|)^{\calO(d\cdot d!)}$ time.
  \end{itemize}
\end{theorem}
\end{samepage}

\section{Algorithms for One-Player and Decreasing 
Consumption Games}
\label{sec-restricted}

In this section we present more efficient algorithms for the two 
subclasses of decreasing and one-player consumption games.
Observe that these special classes of games can still retain a rich modeling power.
In particular, the decreasing subclass is quite natural as systems 
that do not decrease some of the resources for a long time most probably 
stopped working completely (also recall that the game considered in
Section~\ref{sec-intro} is decreasing). One-player consumption games 
are  useful for modelling a large variety of scheduling problems, 
as it is illustrated in the following example.

Consider the following (a bit idealized) problem of supplying shops
with goods such as, e.g., bottles of drinking water.  This problem may
be described as follows: Imagine a map with $c$ cities connected by
roads, $n$ of these cities contain shops to be supplied, $k$ cities
contain warehouses with huge amounts of the goods that should be
distributed among the shops. The company distributing the goods owns
$d$ cars, each car has a bounded capacity. The goal is to distribute
the goods from warehouses to all shops in as short time as possible.
This situation can be modeled using a one-player consumption game as
follows. States would be tuples of the form $(c_1,\ldots,c_d,A)$ where
each $c_i \in \{1,\ldots,c\}$ corresponds to the city in which the
$i$-th car is currently located, $A\subseteq \{1,\ldots,n\}$ lists the
shops that have already been supplied (initially $A=\emptyset$ and the
goal is to reach $A=\{1,\ldots,n\}$).  Loads of individual cars and
the total time would be modelled by a vector of resources,
$(\ell(1),\ldots,\ell(d),t)$, where each $\ell(i)$ models the current
load of the $i$-th car and $t$ models the amount of time which elapsed
from the beginning (this resource is steadily decreased until
$A=\{1,...,n\}$). Player $\Box$ chooses where each car should go next.
Whenever the $i$-th car visits a city with a warehouse, the
corresponding resource $\ell(i)$ may be reloaded. Whenever the $i$-th
car visits a city containing a shop, player $\Box$ may choose to
supply the shop, i.e. decrease the resource $\ell(i)$ of the car by
the amount demanded by the shop.
Now the last component of a minimal safe configuration indicates how
much time is needed to supply all shops.  A cover configuration
indicates not only how much time is needed but also how large cars are
needed to supply all shops.  This model can be further extended with
an information about the fuel spent by the individual cars, etc.

As in the previous section, we fix a consumption game
$\C = (S,E,(S_{\Box},S_{\Diamond}),L)$ of dimension~$d$, and we use
$\ell$ to denote the maximal finite $|\delta(i)|$ such that 
\mbox{$1 \leq i \leq d$} and $\delta$ is a label of some transition.
We first establish the complexity of emptiness and membership 
problem, and then present an algorithm to compute the minimal safe 
configurations.

\subsection{The Emptiness and Membership Problems} 
\label{sec:restricted-bounded-prop}
We first establish the complexity of the emptiness problem for 
decreasing games by a polynomial time reduction to \emph{generalized B\"uchi games}.
A generalized B\"uchi game is a tuple 
\mbox{$\B = (V,\gtran,(V_{\Box},V_\Diamond),B)$}, where 
$(V,\gtran,(V_{\Box},V_\Diamond))$ is a 2-player game with finitely many
vertices, and $B = \{F_1,\ldots,F_m\}$, where $m \geq 1$
and $F_i \subseteq {\gtran}$ for all $1 \leq i \leq m$.
We say that infinite path $w$ \emph{satisfies} the generalized B\"uchi 
condition defined by $B$ iff $\inf(w) \cap F_i \neq \emptyset$ for every
$1 \leq i \leq m$. 
A strategy $\sigma \in \Sigma_{\B}$ is \emph{winning} in $v \in V$ 
if for every $\pi \in \Pi_{\B}$ we have that the $\play_{\sigma,\pi}(v)$ 
satisfies the generalized B\"uchi condition. 
The problem whether player $\Box$ has a winning strategy in state $s$ can be 
decided in polynomial time, with an algorithm of complexity 
$\calO(|V|\cdot |\gtran|\cdot m)$ (see~\cite{DJW97}).

We claim that the following holds:

\begin{lemma}
\label{lem:decreasing-gbuchi}
If $\C$ is a decreasing game then $\safe(s)\neq\emptyset$ if and only if the player $\Box$ has 
winning strategy in generalized B\"uchi game $\B_{\C} = 
(S,E,(S_{\Box},S_{\Diamond}),\{R_1,\dots,R_d\})$ where for each $1 \leq i \leq d$ we have $R_i = $ \mbox{$\{(s,t)~\in~E~\mid L(s,t)(i)~=~\omega\}$}.
\end{lemma} %

\noindent
Previous lemma immediately gives us that the emptiness of $\safe(s)$ in 
decreasing games is decidable in time $\calO(|S|\cdot|E|\cdot d)$.
We now argue that the emptiness of $\safe(s)$ for one-player games 
can also be achieved in polynomial time.
Note that from Lemma~\ref{lem-con-Street} we have that $\safe(s)\neq\emptyset$ 
if and only if player~$\Box$ has winning strategy in one-player Streett game $\calS_{\C}$. 
The problem of deciding the existence of winning strategy in one-player Streett game 
is exactly the nonemptiness problem for Streett automata that can be solved in time 
$\calO((|S|\cdot d+|E|)\cdot \min\{|S|,d\})$~\cite{KurshanBook}.

\begin{theorem}\label{thrm:dec-one-emptiness}
Given a consumption game $\C$ and a state $s$, 
the emptiness problems of whether $\safe(s)=\emptyset$ and 
$\cover(s)=\emptyset$ can be decided in time 
$\calO(|S|\cdot |E| \cdot d)$ if $\C$ is decreasing,
and in time $\calO((|S|\cdot d+|E|)\cdot \min\{|S|,d\})$
if $\C$ is a one-player game.
\end{theorem}

\noindent
We now study the complexity of the membership problem for $\safe(s)$.
We prove two key lemmas and the lemmas bound the number of steps 
before all resources are reloaded.
The key idea is to make player~$\Box$ reload resources as soon as possible.
Formally, we say that a play $\play_{\sigma,\pi}(s,\alpha)$ induced by a 
sequence of transitions $s_0\tran{\delta_1} \cdots \tran{\delta_k} s_k$ \emph{reloads $i$-th resource in $j$-th step} 
if $\delta_j(i)=\omega$. 
We first present a lemma for decreasing games and then for one-player games.

\begin{lemma}\label{lem-claim1}
Consider a decreasing consumption game $\C$ and 
a configuration $(s,\alpha)$ such that $\alpha\in \safe(s)$.
There is a safe strategy $\sigma$ of player $\Box$ in $(s,\alpha)$ such that every 
$\play_{\sigma,\pi}(s,\alpha)$ reloads all resources in the first $d\cdot |S|$ steps.
\end{lemma}

\noindent
Now let us consider {\it one-player} games. As player $\Diamond$ has only one trivial strategy, $\pi$, we write only $\play_{\sigma}(s,\alpha)$ instead of $\play_{\sigma,\pi}(s,\alpha)$.
\begin{lemma}\label{lem-claim2}
Consider a one-player consumption game $\C$ and 
a configuration $(s,\alpha)$ such that $\alpha\in \safe(s)$.
There is a safe strategy $\sigma$ of player $\Box$ in $(s,\alpha)$ such that for the $\play_{\sigma}(s,\alpha)$ and every $1\leq i\leq d$ we have that either the $i$-th resource is reloaded in the first $d\cdot |S|$ steps, or it is never decreased from the $(d\cdot |S|+1)$-st step on.
\end{lemma}

\noindent
As a consequence of Lemma~\ref{lem-claim1}, Lemma~\ref{lem-claim2} 
and the hardness results presented in Proposition~\ref{prop:hard-main} (in appendix) 
we obtain the following:
\begin{theorem}\label{thrm:dec-one-membership}
The membership problem of whether  $\alpha \in \safe(s)$ is $\NP$-complete for one-player consumption games and 
$\PSPACE$-complete for decreasing consumption games. The problem whether $\alpha$ is a minimal element of $\safe(s)$ is $\DP$-complete for one-player consumption games and 
$\PSPACE$-complete for decreasing consumption games.
\end{theorem}

\subsection{Minimal Safe Configurations and Multi-Distance Reachability}
\label{sec:reachability}

In the rest of the paper we present algorithms for computing the minimal safe configurations in one-player and decreasing
consumption games. Both algorithms use the iterative algorithm for \emph{multi-distance reachability problem}, which is described below, as a subprocedure. Although their worst-case complexity is the same as the complexity of generic algorithm from Section \ref{sec-general}, we still deem them to be more suitable for practical computation due to some of their properties that we state here in advance:
\begin{itemize}
 \item The generic algorithm always constructs game of size $(|S|\cdot d \cdot \ell)^{\calO(d)}$. In contrast, algorithms based on solving multi-distance reachability construct a game whose size is linear in size of $\C$ \emph{for every fixed choice of parameter} $d$.
 \item The multi-distance reachability algorithms iteratively construct sets of configurations that are safe but may not
be minimal before the algorithm stops. Although the time complexity of this iterative computation is $(|S|\cdot d \cdot \ell)^{\calO(d)}$ at worst, it may be the case that the computation terminates much earlier. Thus, these algorithms have a chance to terminate earlier than in $(|S|\cdot d \cdot \ell)^{\calO(d)}$ steps (unlike the generic algorithm, where the necessary construction of the ``large'' safety game always requires this number of steps).
 \item Moreover, the algorithm for one-player games presented in Section \ref{subsec-oneplayer} decomposes the problem into many parallel subtasks that can be processed
independently.
\end{itemize}

Let $\D$ denote a $d$-dimensional consumption game with
transitions labeled by vectors over $\Zset_{\leq 0}$ (i.e. there is no
$\omega$ in any label). Also denote $D$ the set of states of game $\D$. We say that vector $\alpha$ is
a \emph{safe multi-distance} (or just \emph{safe distance}) from state $s$ to state $r$ if there is a strategy $\sigma$ for player $\Box$ such that for any strategy $\pi$ for
player $\Diamond$ the infinite path $\play_{\sigma,\pi}(s,\alpha)$
visits a configuration of the form $(r,\beta)$. That is, $\alpha$ is a
safe distance from $s$ to $r$ if player $\Box$ can enforce reaching
$r$ from $s$ in such a way that the total decrease in resource values
is less than $\alpha$.

We denote by $\safe_{\D}(s,r)$ the set of all safe distances from $s$
to $r$ in $\D$, and by $\lambda_{\D}(s,r)$ the set of all minimal
elements of $\safe_{\D}(s,r)$. If $\safe_{\D}(s,r) = \emptyset$, then
we set $\lambda_{\D}(s,r) = (\infty,\dots,\infty)$, where the symbol
$\infty$ is treated accordingly with the usual conventions (for any $c
\in \Zset$ we have $\infty-c = \infty$, $c < \infty$; we do not use
the $\omega$ symbol to avoid confusions).

We present a simple fixed-point iterative algorithm which computes the
set of minimal safe distances from $s$ to $r$. Apart from the standard set
operations, the algorithm uses the following operations on sets of 
vectors: for a given set $M$ and a given vector $\alpha$, the 
operation {\tt min-set}$(M)$
returns the set of minimal elements of $M$, and $M - \alpha$ returns the set
$\{\beta - \alpha\mid \beta \in M\}$. Further, given a sequence of sets of
vectors $M_1,\dots,M_m$ the operation {\tt cwm}$(M_1,\dots,M_m)$
returns the set $\{\alpha_1 \vee \dots \vee \alpha_m \mid \alpha_1 \in
M_1,\dots,\alpha_m \in M_m\}$, where each $\alpha_1 \vee \dots \vee
\alpha_m$ denotes a component-wise maximum of the vectors $\alpha_1,
\dots, \alpha_m$.
Technically, the algorithm iteratively solves the following optimality
equations: for any state $q$ with outgoing transitions
$q\ltran{\delta_1} q_1,\dots,q\ltran{\delta_m} q_m$ we have that
\[
 \lambda_{\D}(q,r) = \begin{cases}
                    \text{{\tt min-set}}\left(\lambda_{\D}(q_1,r) -\delta_{1} \cup \dots \cup \lambda_{\D}(q_m,r) - \delta_{m}\right) & \text{if } q \in D_{\Box} \\
		    \text{{\tt min-set}}\left(\ \text{{\tt cwm}}(\lambda_{\D}(q_1,r) - \delta_1,\dots,\lambda_{\D}(q_m,r) - \delta_{m})\ \right) & \text {if } q \in D_{\Diamond}
                    \end{cases}
\]

The algorithm iteratively computes the $k$-step approximations of
$\lambda_{\D}(q,r)$, which are denoted by $\lambda^k_{\D}(q,r)$.  Intuitively,
each set $\lambda^k_{\D}(q,r)$ consists of all minimal safe distances
from $q$ to $r$ over all plays with at most $k$ steps. The set
$\lambda^0_{\D}(q,r)$ is initialized to $\{(\infty,\ldots,\infty)\}$
for $q\not = r$, and to $\{(1,\ldots,1)\}$ for $q=r$. Each
$\lambda^{k+1}_{\D}(q,r)$ is computed from $\lambda^{k}_{\D}(q,r)$
using the above optimality equations until a fixed point is reached.
In Appendix~\ref{app-restricted} we show that this fixed point is the
correct solution for the minimal multi-distance problem.

Since the algorithm is based on standard methods, we postpone its 
presentation to Appendix~\ref{app-restricted} and state only 
the final result.
We 
call \emph{branching degree} of $\D$ the maximal number of transitions outgoing from any state of $\D$.
\begin{theorem}
 \label{thm:min-dist}
There is an iterative procedure Min-dist$(\D,s,r)$ that correctly computes the set of minimal safe distances from $s$ to $r$ in time $\calO\left(|D| \cdot a \cdot b\cdot N^2\right)$, where $b$ is the branching degree of $\D$, $a$ is the length of a longest acyclic path in $\D$ and $N=\max_{0 \leq k \leq a}|\lambda^k_{\D}(q,r)|$. 

Moreover, the procedure requires at most $a$ iterations to converge to the correct solution and thus
the resulting set $\lambda_{\D}(s,r)$ has size at most $N$. Finally, the number $N$ can be bounded from above by $(a\cdot \ell)^d$.
\end{theorem}

\noindent
Note that the complexity of the procedure Min-dist$(\D,s,r)$ crucially depends on parameter $N$. The bound on $N$ presented in the previous theorem follows from the obvious fact that components of all vectors in $\lambda_{\D}^{k}(s,r)$ are either all equal to $\infty$ or are all bounded from above by $k\cdot \ell$. However, for concrete instances the value of $N$ can be substantially smaller. For example, if the consumption game $\D$ models some
 real-world problem, then it can be expected that the number of $k$-step minimal
distances from states of $\D$ to $r$ is small, because changes in resources
 are not entirely independent in these models (e.g., action that
 consumes a large amount of some resource may consume a large
 amount of some other resources as well). This observation forms the core of our claim that algorithms based on multi-distance reachability may terminate much earlier than the generic algorithm from Section \ref{sec-general}.

\subsection{Computing $\safe(s)$ in One-Player Consumption Games}\label{subsec-oneplayer}

Now we present an algorithm for computing minimal elements of $\safe(s)$ in one-player consumption games. The algorithm computes the solution by solving several instances of minimum multi-distance reachability problem.
We assume that all states $s$ with $\safe(s)=\emptyset$
were removed from the game. This can be done in polynomial time using
the algorithm for emptiness
(see Theorem~\ref{thrm:dec-one-emptiness}).

We denote by $\Pi(d)$ the set of all permutations of the set
$\{1,\dots,d\}$. We view each element of $\Pi(d)$ as a finite sequence
$\pi_1 \dots \pi_d$, e.g., $\Pi(2)=\{12,21\}$.
We use the standard notation $\pi$ for permutations: confusion with
strategies of player $\Diamond$ should not arise since~$S_\Diamond = \emptyset$.

We say that a play $\play_{\sigma}(s,\alpha)$ {\em matches} a
permutation $\pi$ if for every $1 \leq i<j \leq d$ the following
holds: If both $\pi_i$-th and $\pi_j$-th resources are reloaded during
$\play_{\sigma}(s,\alpha)$, then the first reload of $\pi_i$-th resource occurs
before or at the same time as the first reload of $\pi_j$-th resource.
A configuration $(s,\alpha)$ matches $\pi$ if there is a strategy
$\sigma$ that is safe in $(s,\alpha)$ and $\play_{\sigma}(s,\alpha)$
matches $\pi$.  We denote by $\safe(s,\pi)$ the set of all vectors
$\alpha$ such that $(s,\alpha)$ matches $\pi$.  Note that
$\safe(s)=\bigcup_{\pi\in\Pi(d)} \safe(s,\pi)$.

As indicated by the above equality, computation of safe configurations
in $\C$ reduces to the problem of computing, for every permutation
$\pi$, safe configurations that match~$\pi$. The latter problem, in
turn, easily reduces to the problem of computing safe multi-distances in
specific one-player consumption games $\C(\pi)$.  Intuitively, each game $\C(\pi)$
simulates the game $\C$ where the resources are forced to be reloaded
in the order specified by $\pi$. So the states of each $\C(\pi)$ are
pairs $(s,k)$ where $s$ corresponds to the current state of the
original game and $k$ indicates that the first $k$ resources, in the
permutation $\pi$, have already been reloaded.  Now the crucial point
is that if the first $k$ resources have been reloaded  when some configuration $c=(s,\beta)$ of the original game is visited, and there is a
safe strategy in $c$ which does not decrease any of the resources with the index
greater than $k$, then we may safely conclude that the {\em initial}
configuration is safe. So, in such a case we put a transition from the
state $(s,k)$ of $\C(\pi)$ to a distinguished target state $r$ (whether or not to put in
such a transition can be decided in polynomial time due to
Theorem~\ref{thrm:dec-one-emptiness}). Other transitions of $\C(\pi)$
correspond to transitions of $\C$ except that they have to update the
information about already reloaded resources, cannot skip any resource
in the permutation (such transitions are removed), and the components
indexed by
$\pi_1,\dots,\pi_k$ are substituted with $0$ in transitions incoming
to states of the form $(q,k)$ (since already reloaded resources become
unimportant as indicated by the above observation).

A complete construction of $\C(\pi)$ is presented in 
Appendix~\ref{app-restricted} as a part of a formal proof of the 
following theorem:

\begin{theorem}
\label{thm:oneplayer-reach}
For every permutation $\pi$ there is a polynomial time constructible
consumption game $\C(\pi)$ of size $\calO(|S|\cdot d)$ and branching degree $\calO(|S|)$ such that for
every vector $\alpha$ we have that $\alpha \in \safe(s,\pi)$ in $\C$ iff
$\alpha$ is a safe distance from $(s,0)$ to $r$ in
$\C(\pi)$. %
\end{theorem}

\noindent
By the previous theorem, every minimal element of $\safe(s)$ is an
element of $\lambda_{C({\pi})}((s,0),r)$ for at least one permutation
$\pi$. Our algorithm examines all permutations $\pi \in \Pi(d)$, and for
every permutation it constructs game $C({\pi})$ and computes
$\lambda_{C({\pi})}((s,0),r)$ using the procedure Min-dist from
Theorem \ref{thm:min-dist}. The algorithm also stores the set of all
minimal vectors that appear in some $\lambda_{C({\pi})}((s,0),r)$. In this way, the algorithm eventually finds all minimal elements of $\safe(s)$. The
pseudocode of the algorithm can be found in Appendix~\ref{app-restricted}.

From complexity bounds of Theorems~\ref{thrm:dec-one-emptiness} and~\ref{thm:min-dist} we obtain that the worst case running time of this algorithm is $d!\cdot(|S|\cdot\ell \cdot d)^{\calO(d)}$.
In contrast with the generic algorithm of
Section~\ref{sec-general}, that constructs an exponentially large safety game, the
algorithm of this
section computes $d!$ ``small'' instances of the minimal multi-distance
reachability problem. We can solve many of these instances in
parallel. Moreover, as argued in previous section, each call
of Min-dist$(\C({\pi}),(s,0),r)$ may have much better running time
than the worst-case upper bound suggests.
\newcommand{\wh}{\widehat}

\subsection{Computing $\safe(s)$ in Decreasing Consumption Games}\label{subsec-decreasing}

We now turn our attention to computing minimal elements of $\safe(s)$ in decreasing 
games.
The main idea is again to reduce this task to the computation of minimal 
multi-distances in a certain consumption game. We again assume that states with $\safe(s)=\emptyset$ were removed from the game. 

The core of the reduction is the following observation: if $\C$ is 
decreasing, then $\alpha\in\safe(s)$ iff player~$\Box$ is 
able to ensure that the play satisfies these two conditions: all 
resources are reloaded somewhere along the play; and the 
$i$-th resource is reloaded for the first time before it is decreased by at least 
$\alpha(i)$. %
Now if we %
augment the states 
of $\C$ with an information about which resources have been reloaded at least once in 
previous steps, then the objective of player~$\Box$ is actually to reach a 
state which tells us that all resources were reloaded at least once. 

So the algorithm constructs a game $\wh{\C}$ by augmenting states of
$\C$ with an information about which resources have been reloaded at
least once, and by substituting updates of already reloaded resources
(i.e., the corresponding components of the labels) with zeros. Note
that the construction of $\wh{\C}$ closely resembles the construction
of games $\C(\pi)$ from the previous section. However, in two-player
case we cannot fix an order in which resources are to be reloaded,
because the optimal order depends on a strategy chosen by player
$\Diamond$. Thus, we need to remember exactly which resources have been
reloaded in the past (we only need to remember the set of resources that 
have been reloaded, but not the order in which they were reloaded).

We leave the formal construction of game $\wh{C}$ to 
Appendix~\ref{app-restricted} together with a proof of the following theorem which only states the final reults.

\begin{theorem}
\label{thm:decreasing}
There is a consumption game $\wh{\C}$ of size $\calO(2^d \cdot |S|)$, branching degree $\calO(S)$ and with maximal acyclic path of length $\calO(|S|\cdot d)$, with the following properties: $\wh{\C}$ is constructible in time $\calO(2^d \cdot (|S|+|E|))$ and for every vector $\alpha$ we have $\alpha\in\safe(s)$ in $\C$ iff $\alpha$ is a safe distance from $(s,\emptyset)$ to $r$ in $\wh{C}$. 
\end{theorem}
\noindent
The previous theorem shows that we can find minimal elements of $\safe(s)$ with a single call of procedure Min-dist$(\wh{\C},(s,\emptyset),r)$. Straightforward complexity analysis reveals that the worst-case running time of this algorithm is $(|S|\cdot d \cdot \ell )^{\calO(d)}$.
However, the game $\wh{\C}$ constructed during the computation is still smaller than the safety game
constructed by the generic algorithm of
Section~\ref{sec-general}. %
Moreover, the length of the longest acyclic path in $\wh{\C}$ is bounded by
$|S|\cdot d$, so the procedure \mbox{Min-dist} does not have to perform many
iterations, despite the exponential size of $\wh{\C}$. Finally, let us once again recall that the procedure Min-dist$(\wh{\C},(s,\emptyset),r)$ may actually require much less than $(|S|\cdot d \cdot \ell )^{\calO(d)}$ steps.

\section{Conclusions}
\label{sec-concl}

As it is witnessed by the results presented in previous sections,
consumption games represent a convenient trade-off between 
expressive power and computational tractability. The presented
theory obviously needs further development before it is implemented
in working software tools. Some of the issues are not yet fully
understood, and there are also other well-motivated problems
about consumption games which were not considered in this paper. 
The list of important open problems includes the following:
\begin{itemize}
\item Improve the complexity of algorithms for $\cover(s)$. This
  requires further insights into the structure of these sets.
\item Find efficient controller synthesis algorithms for objectives
  that combine safety with other linear-time properties. That is, decide
  whether player~$\Box$ has a safe strategy such that a play satisfies
  a given LTL property no matter what player~$\Diamond$ does.
\item Find algorithms for more complicated optimization problems, where
  the individual resources may have different priorities. For example,
  it may happen that fuel consumption or the price of batteries with 
  large capacity are much more important than the time spent,
  and in that case we might want to optimize some weight function
  over the tuple of all resources. It may happen (and we have concrete
  examples) that some of these problems are actually solvable even more
  efficiently than the general ones where all resources are treated 
  equally w.r.t.\ their importance.
\end{itemize} 

\noindent
The above list is surely incomplete. The problem of optimal resource 
consumption is rather generic and appears in many different contexts,
which may generate other interesting questions about consumption 
games.

\bibliographystyle{splncs03}
\bibliography{main_cav}

\newpage
\appendix

\noindent
\begin{center}
\huge\bf Technical Appendix
\end{center}
\section{Proofs about Streett Games}
\label{app-Street}

Let  $\C = (S,E,(S_{\Box},S_{\Diamond}),L)$ be a consumption game 
of dimension~$d$.

\subsection*{Proof of Lemma~\ref{lem-con-Street}}
First we observe that every strategy for the players in the Streett game 
is a strategy in the original consumption game and vice versa.
\begin{enumerate}
\item To prove the result of the first item, we show the contrapositive,
i.e., we show that if player~$\Box$ does not have a winning strategy 
in $s$ in the Streett game, then $\safe(s)=\emptyset$. 
If player~$\Box$ does not have a winning strategy, then by the determinacy
of Streett games and by existence of memoryless winning strategies for 
Rabin winning condition (complement of Streett winning condition)~\cite{EJ88}, 
there exists a memoryless winning strategy $\pi$ for player~$\Diamond$.
Fix the memoryless winning strategy $\pi$ and consider the graph 
$\calS_{\C}^\pi$ obtained by fixing the strategy $\pi$ for player $\Diamond$ (in $\calS_{\C}^{\pi}$ 
only player~$\Box$ makes choices).
Consider any strongly connected component 
(scc) $U$ reachable from $s$ in $\calS_{\C}^{\pi}$ (we denote by $U$ the set 
of edges of the scc). 
Then there must exist $1 \leq i \leq d$ such that $U \cap G_i\neq \emptyset$ 
and $U \cap R_i=\emptyset$; otherwise, against $\pi$, player~$\Box$ can execute infinitely often
exactly the edges of $U$ and ensure the Streett winning condition,
contradicting $\pi$ is a winning strategy for player $\Diamond$.
Consider the strategy $\pi$ in the original consumption game, and  
arbitrary strategy $\sigma$ for player~$\Box$.
Let $w=\play_{\sigma,\pi}(s)$ be the play given $\sigma$ and $\pi$ from $s$, 
and let $U'=\inf(w)$.
Then $U'$ is a scc in $\calS_{\C}^\pi$, and hence there exists 
$1 \leq i \leq d$ such that $U' \cap G_i\neq \emptyset$ and 
$U' \cap R_i=\emptyset$, i.e., in dimension $i$ there exists at least one 
transition $e_1$ where the weight
is negative and there is no transition where the weight is positive.
Since $e_1$ is executed infinitely often and there is no increase in dimension 
$i$ from some point on (as the set executed infinitely often is $U'$ with no 
$\omega$-weight transition in dimesion $i$), it follows that irrespective of 
the initial vector $\alpha$ we have $(s,\alpha) \not\in \safe(s)$.
This completes the first item of the result.

\item We now show the second item. 
If there is a winning strategy in the Streett game in $s$, 
then it follows from the results of~\cite{Zie98,DJW97} that there
exists a winning strategy $\sigma$ of memory size at most $d!$.
Fix the winning strategy $\sigma$ of memory size $d!$, and 
consider the graph $\calS_{\C}^{\sigma}$ obtained by fixing $\sigma$ 
(it is the synchronous product obtained from $\calS_{\C}$ and the 
deterministic transducer of size $d!$ describing $\sigma$).
In $\calS_{\C}^\sigma$ only player~$\Diamond$ makes choices, and 
the size of the vertex set of the graph is $d! \cdot |S|$.
Consider a cycle $U$ reachable from $(s,m_0)$ in $\calS_{\C}^\sigma$,
where $m_0$ is the initial memory state.
The cycle $U$ must satisfy that for all $1 \leq i \leq d$ if 
$U \cap G_i \neq \emptyset$, then $U \cap R_i\neq \emptyset$; 
otherwise visiting exactly the cycle $U$ infinitely often against 
$\sigma$ player~$\Diamond$ can contradict that $\sigma$ is a winning 
strategy.
It follows that given the strategy $\sigma$, given any arbitrary 
counter-strategy $\pi$ for player~$\Diamond$, for all $1 \leq i \leq d$, 
within every $d! \cdot |S|$ visits of $G_i$ there is at least one visit to 
$R_i$.
In other words, in the consumption game, the strategy $\sigma$ ensures that 
for every dimension $i$, within $d! \cdot |S|$ visits of negative weight 
transitions there is at least one $\omega$-weight transition.
Since the maximum negative weight is at most $\ell$, it follows that starting
with weight vector 
$(d! \cdot |S|\cdot \ell +1 ,\ldots,d! \cdot |S|\cdot \ell +1)$,
and whenever a $\omega$-weight transition is visited in dimension $i$ 
reloading upto value $d!\cdot |S| \cdot \ell +1$, ensures that 
$(d! \cdot |S|\cdot \ell +1 ,\ldots,d! \cdot |S|\cdot \ell +1) 
\in \safe(s) \cap \cover(s)$.
\end{enumerate}
The desired result follows.

\subsection*{Proof of Lemma~\ref{lem-Street-con}}
 Consider the Streett game $(\C_{\calS})_{\C}$ from Lemma \ref{lem-con-Street} that corresponds to the consumption game $\C_{\calS}$. It is easy to see that $(\C_{\calS})_{\C}=\C$. The rest immediately follows from Lemma~\ref{lem-con-Street}.

\section{Proofs Based on Analyzing Abstract Load Vectors}
\label{app-abstract-load}

In this section we present full proofs of Lemma~\ref{lem-safe-min-bounded}
and Lemma~\ref{lem-cover-min-bounded}. These proofs are based on 
similar underlying ideas, inspired by techniques originally
presented in~\cite{BJK}.

For the rest of this section, we fix a consumption game
$\C = (S,E,(S_{\Box},S_{\Diamond}),L)$ of dimension~$d$, and we use
$\ell$ to denote the maximal finite $|\delta(i)|$ such that 
$1 \leq i \leq d$ and $\delta$ is a label of some transition.

Let us recall the notions introduced in Section~\ref{sec-general}.
An \emph{abstract load vector} $\mu$ is an element of $(\Zset_{>0}^\omega)^d$. 
For every $1 \leq i \leq d$, we say that $\mu(i)$ is \emph{precise}
if $\mu(i) \neq \omega$. The \emph{precision} of $\mu$ is the
number of its precise components. The \emph{type}
of~$\mu$ is the set $T$ of all indexes $i$ such that $\mu(i) \neq \omega$.
For every abstract load vector $\mu$, every
$1 \leq j \leq d$, and every $k \in \Zset_{>0}^{\omega}$, we use
$\mu[j/k]$ to denote the abstract load vector obtained from $\mu$
by replacing its $j$-th component with~$k$. Further, for every 
$\alpha \in (\Zset_{>0})^d$ and every type $T$ we define the
\emph{corresponding} abstract load vector $\mu_\alpha$ of type $T$
where $\mu_\alpha(m) = \alpha(m)$ for all $m \in T$. We also use
$\vec{1}_j$, where $1 \leq j \leq d$, to denote the vector whose
$j$-th component is equal to~$1$ and the other components 
are equal to~$0$.

The standard componentwise ordering is extended also to abstract 
load vectors by stipulating that $c < \omega$
for every $c \in \Zset$. Given an abstract load vector~$\mu$
and a vector $\alpha \in (\Zset_{>0})^d$, we say that $\alpha$ 
\emph{matches $\mu$} if $\alpha(i) = \mu(i)$ for all 
$1 \leq i \leq d$ such that $\mu(i) \neq \omega$. Further, we say that 
$\mu$ is \emph{compatible} with $\safe(s)$ (or $\cover(s)$)
if there is some $\alpha \in \safe(s)$  (or $\alpha \in \cover(s)$) 
that matches~$\mu$. The set of all abstract load vectors that 
are compatible with $\safe(s)$ (or $\cover(s)$) is denoted by
$\scomp(s)$ (or $\ccomp(s)$).

Let $\mu \in \scomp(s)$, and
$\alpha \in \safe(s)$ a vector that matches $\mu$. 
Let $K \in \Zset_{>0}$. We say that a 
strategy $\sigma$ which is safe in $(s,\alpha)$ \emph{stays above $K$
with respect to $\mu$} if for every strategy $\pi$ of 
player~$\Diamond$ and every configuration $(t,\beta)$ visited along
$\play_{\sigma,\pi}(s,\alpha)$ we have that $\beta(i) \geq K$
for every $1 \leq i \leq d$ such that $\mu(i) = \omega$.
The following lemma is immediate.

\begin{lemma}
\label{lem-K}
   Let $\mu \in \scomp(s)$.
   Then for every $K \in \Zset_{>0}$ there is a vector $\alpha \in
   \safe(s)$ and a safe strategy $\sigma$ for $(s,\alpha)$ such that 
   $\alpha$ matches~$\mu$ and $\sigma$ stays above $K$
   with respect to $\mu$.
\end{lemma}

\noindent
Now we have all the tools needed to prove Lemma~\ref{lem-safe-min-bounded}.

\subsection*{Proof of Lemma~\ref{lem-safe-min-bounded}}
We show that for all $0 \leq i \leq d$ and $s\in S$, the precise components
of all \emph{minimal} $\mu \in \scomp(s)$ with precision~$i$
are bounded by $i \cdot \ell \cdot |S|$.
This clearly suffices, because the minimal
$\mu \in \scomp(s)$ with precision~$d$ are exactly the minimal
elements of $\safe(s)$.

We proceed by induction on~$i$. The case when $i=0$ is immediate,
because the only abstract vector with precision~$0$ is
$(\omega,\ldots,\omega)$, and the claim holds trivially. Now
assume that the claim holds for some~$i<d$. Let us fix some 
minimal  $\mu \in \scomp(s)$ with precision~$i$. Let $T$ be 
the type of $\mu$, and let $j \not\in T$, $1 \leq j \leq d$, 
be an arbitrary (but fixed) index. Further, let $k$ be the
\emph{least} number such that $\mu[j/k] \in \scomp(s)$, and
let $K \in \Zset_{>0}$ be (some) number such that 
\begin{itemize}
\item for every $t \in S$ and every $\mu' \in \scomp(t)$ of type 
  $T \cup \{j\}$ such that $\mu' \leq \mu[j/k]$ we have that the 
  vector obtained  from $\mu'$ by substituting every $\omega$-component 
  with~$K$ belongs to $\safe(t)$;
\item for every  $t \in S$ and every \emph{minimal} $\mu' \in \scomp(t)$
  of precision~$i$ we have that the 
  vector obtained  from $\mu'$ by substituting every $\omega$-component 
  with~$K$ belongs to $\safe(t)$.
\end{itemize}
Such a $K$~clearly exists because the total number of all such 
$\mu'$ is finite.

We claim that $k \leq (i+1) \cdot \ell \cdot |S|$. Assume the converse.
We show that then $\mu[j/k{-}1] \in \scomp(s)$, which contradicts the
minimality of~$k$.

Since $\mu[j/k] \in \scomp(s)$, due to Lemma~\ref{lem-K}
there exists $\alpha \in \safe(s)$ and a strategy $\sigma$
for player~$\Box$ such that 
\begin{itemize}
\item $\alpha$ matches $\mu[j/k]$,
\item $\sigma$ is safe in $(s,\alpha)$ and stays above $K$
  with respect to $T \cup \{j\}$.
\end{itemize}
It suffices to show that there is a safe strategy $\sigma'$ 
for player~$\Box$ in $(s,\alpha[j/k{-}1])$. The strategy $\sigma'$ 
keeps mimicking the moves of $\sigma$ until one of the following
conditions is satisfied:
\begin{enumerate}
  \item[C1.] The play visits a configuration $(t,\gamma)$ such that
     the corresponding abstract load vector $\mu_\gamma$ for type 
     $T \cup\{j\}$ satisfies $\mu_\gamma \in \scomp(t)$ and 
     $\mu_\gamma \leq \mu[j/k]$.
  \item[C2.] The play visits a configuration $(t,\gamma)$ such that
     $t \in S_\Box$ and there is a transition 
     $(t,\gamma) \tran{\delta} (v,\beta)$
     such that $\delta(m) = \omega$ for some $m \in T \cup \{j\}$ and 
     the configuration $(v,\beta+\vec{1}_j)$ is safe.
\end{enumerate}
First we show that if player~$\Box$ follows the strategy $\sigma'$ from
$(s,\alpha[j/k{-}1])$, then the \mbox{$j$-th} resource cannot be decremented
by more than $|S|-1$ times along a play unless C1 or C2 happens.
Assume the converse, i.e., player~$\Diamond$ has a strategy $\pi$
such that the $j$-th resource is decremented $|S|$ 
times along $\play_{\sigma',\pi}(s,\alpha[j/k{-}1])$ without encountering
C1 or C2. Then there must be two configurations of the form
$(t,\gamma)$ and $(t,\varrho)$ such that for every $m \in T \cup \{j\}$
we have that $\mu[j/k](m) \geq \gamma(m) \geq \varrho(m)$ and
$\gamma(j) > \varrho(j)$. We show that $(t,\gamma)$ satisfies C1.
Clearly, $\mu_\gamma \leq \mu[j/k]$. It remains to show that
$\mu_\gamma \in \scomp(t)$. This is achieved by identifying
another abstract load vector $\mu'$ of type $T \cup \{j\}$ such that
$\mu' \leq \mu_\gamma$ and $\mu' \in \scomp(t)$.  Since $\sigma'$ mimics 
$\sigma$ and $\sigma$ is safe in $(s,\alpha[j/k])$, we have that 
$(t,\varrho+\vec{1}_j)$ is safe. However, this means that the corresponding
vector $\mu_{\varrho+\vec{1}_j}$ for type $T \cup\{j\}$ satisfies
$\mu_{\varrho+\vec{1}_j} \in \scomp(t)$, and we also have that 
$\mu_{\varrho+\vec{1}_j} \leq \mu_\gamma$ as required.

Obviously, if $C1$ happens, then the configuration $(t,\gamma)$ is safe
because $\sigma$ stays above~$K$ and $\sigma'$ mimics $\sigma$. Hence,
player~$\Box$ can simply switch to a safe strategy for $(t,\gamma)$.
If $C2$ happens before $C1$, i.e., the play visits a configuration
of $(t,\gamma)$ such that $t \in S_\Box$ and there is a transition 
$(t,\gamma) \tran{\delta} (v,\beta)$
such that $\delta(m) = \omega$ for some $m \in T \cup \{j\}$ and 
the configuration $(v,\beta+\vec{1}_j)$ is safe,
then $\beta(j) \geq i \cdot \ell \cdot |S|$, because the \mbox{$j$-th} 
resource could not be decremented by more than $|S|-1$ transitions so far.
Now consider the corresponding abstract load vector $\mu_{\beta+\vec{1}_j}$ for 
type $(T \cup \{j\}) \smallsetminus \{m\}$. Since
$\mu_{\beta+\vec{1}_j} \in \scomp(v)$ and its precision is $i$, 
we can apply induction 
hypothesis and conclude that there is another $\mu' \in \scomp(v)$
of the same type $(T \cup \{j\}) \smallsetminus \{m\}$ such that
$\mu' \leq \mu_{\beta+\vec{1}_j}$ and all precise components of $\mu'$ are bounded
by $i \cdot \ell \cdot |S|$. However, this means that
$\mu' \leq \mu_\beta$, where $\mu_\beta$ is the abstract load 
vector for type $(T \cup \{j\}) \smallsetminus \{m\}$ corresponding to
$\beta$. Hence, if $\sigma'$ executes the transition 
$(t,\gamma) \tran{\delta} (v,\beta')$ where $\beta'$ is the same as $\beta$
but the $m$-th component is reloaded to at least~$K$, then $(v,\beta')$
is a safe configuration, and $\sigma'$ can switch to a safe strategy for 
$(v,\beta')$.

If $C1$ and $C2$ are not encountered at all, then the $j$-th resource
cannot be decremented by more that $|S|-1$ transitions at all, and hence
$\sigma'$ can mimic $\sigma$ safely forever.

So, we have shown that if $\mu \in \scomp(s)$ is a minimal abstract load
vector of type $T$ with precision~$i$, then for every $j \not\in T$,
$1 \leq j \leq d$, we have that the least $k$ such that 
$\mu[j/k] \in \scomp(s)$ is bounded by  $(i+1) \cdot \ell \cdot |S|$.
However, not all minimal elements of $\scomp(s)$ with precision $i+1$ 
are necessarily of this form. So, let $\mu' \in \scomp(s)$ be some 
(unspecified) minimal abstract load vector of type $T'$ with precision 
$i+1$, and suppose that $\mu'(m) > (i+1) \cdot \ell \cdot |S|$ for some
$m \in T'$. Let $\mu''$ be the abstract load vector of type 
$T' \smallsetminus \{m\}$ obtained from $\mu'$ be replacing $\mu'(m)$
with~$\omega$. Since the precision of $\mu''$ is~$i$ and
$\mu'' \in \scomp(s)$, we can apply induction hypothesis and conclude
that there is a minimal $\hat{\mu} \in \scomp(s)$ of type 
$T' \smallsetminus \{m\}$ such that $\hat{\mu} \leq \mu''$. By the
above presented argument, the vector $\hat{\mu}[m/k]$ of type 
$T'$, where $k = (i+1) \cdot \ell \cdot |S|$,
belongs to $\scomp(s)$. Since $\hat{\mu}[m/k] \leq \mu'$ and the
$m$-component of $\hat{\mu}[m/k]$ is even strictly smaller, 
we have a contradiction with the minimality of $\mu'$.

\subsection*{Proof of Lemma~\ref{lem-cover-min-bounded}}

Now we show how to compute the set of minimal elements of $\cover(s)$
for every $s \in S$. For every $0 \leq i \leq d$, let $B_i \in \Zset_{>0}$
be the least number such that for every $s \in S$, every
minimal abstract load vector $\mu \in \ccomp(s)$ of precision $i$, and
every $\alpha \in \Zset_{>0}^d$ obtained from $\mu$ by substituting
every $\omega$ with $B_i$ we have that $\alpha \in \cover(s)$.

  Similarly as in the proof of Lemma~\ref{lem-safe-min-bounded}, we
  show that for all $0 \leq i \leq d$ and $s\in S$, the precise components
  of all \emph{minimal} $\mu \in \ccomp(s)$ with precision~$i$
  are bounded by $(d \cdot \ell \cdot |S|)^{i!}$. By induction on~$i$.
  The base case when $i = 0$ is immediate. Now
  assume that the claim holds for some~$i<d$. Let us fix some 
  minimal  $\mu \in \ccomp(s)$ with precision~$i$. Let $T$ be 
  the type of $\mu$, and let $j \not\in T$, $1 \leq j \leq d$, 
  be an arbitrary (but fixed) index. Consider the least 
  $k$~such that $\mu[j/k] \in \ccomp(s)$. We argue that 
  $k \leq (d \cdot \ell \cdot |S|)^{(i+1)!}$. To see this, realize
  that we can construct another consumption game $\C_{\mu}$ of dimension
  $d-i$ where the current values of all resources indexed by~$T$
  are encoded explicitly in the states of $\C_{\mu}$. By applying
  induction hypothesis, we can conclude that the number of states
  of $\C_{\mu}$ is bounded by $|S| \cdot (d \cdot \ell \cdot |S|)^{i \cdot i!}$.
  By applying Corollary~\ref{cor-global-bound} to $\C_{\mu}$, we can conclude
  that a sufficient value for $k$ is at most
  \mbox{$(d-i) \cdot \ell \cdot |S| \cdot (d \cdot \ell \cdot |S|)^{i \cdot i!}$},
  which is bounded by $(d \cdot \ell \cdot |S|)^{(i+1)!}$. The rest of the
  argument is the same as in the proof of Lemma~\ref{lem-safe-min-bounded}

\subsection*{Proof of Upper Bounds of Theorem~\ref{thm-cover}}
\noindent
Due to Lemma~\ref{lem-cover-min-bounded}, 
all components of all minimal elements
of $\cover(s)$ are bounded by $(d \cdot \ell \cdot |S|)^{d!}$.
Hence, there are $(d \cdot \ell \cdot |S|)^{d \cdot d!}$ possible candidates 
for the minimal elements of $\cover(s)$, and we can (in principle)
check them by computing a winning region in a safety game of size
$|S| \cdot (d \cdot \ell \cdot |S|)^{d\cdot d!}$. Hence, the set all
minimal elements of all $\cover(s)$, $s \in S$, is computable in time
$(d \cdot \ell \cdot |S|)^{\calO(d\cdot d!)}$. 
Thus, we obtain the
upper complexity bounds of Theorem~\ref{thm-cover}.

\section{Proofs of Lower Bounds}
\label{app-lb}
The purpose of this appendix is to provide proofs of lower bounds from Theorem \ref{thm-safe}, Theorem \ref{thm-cover} and Theorem \ref{thrm:dec-one-membership}. More precisely, we will prove the following proposition:

\tikzstyle{sbox}=[thick,draw,minimum size=1.6em,inner sep=0.1em]
\tikzstyle{sdiam}=[diamond,thick,draw,minimum size=1.8em,inner sep=0.1em]
\tikzstyle{tran}=[thick,draw,->]
\tikzstyle{dummy}=[inner sep=0em]

\begin{proposition}
\label{prop:hard-main}
The problems whether a given vector $\alpha \in \Zset^d_{\geq 0}$ is an element of $\safe(s)$, element of $\cover(s)$, minimal element of $\safe(s)$ or minimal element of $\cover(s)$, in state $s$ of given $d$-dimensional consumption game $\C$ are:
\begin{itemize}
 \item $\NP$-hard for 1-player consumption games.
 \item $\PSPACE$-hard for 2-player consumption games.
\end{itemize}
Moreover the problem whether a given vector $\alpha \in \Zset^d_{\geq 0}$ is a minimal element of $\safe(s)$ is $\DP$-hard for 1-player consumption games. 
Moreover, all of these lower bounds hold even if we restrict ourselves to games where components of all labels are nonzero.
\end{proposition}

In the whole section, for any $n \in \Zset \cup \{\omega\}$ we denote $\vec{n}$ the vector $(n,\dots,n)$.

We will prove Proposition \ref{prop:hard-main} in two stages. First, we will show the following:

\begin{lemma}
\label{lem:np}
 The problems whether a given $\alpha \in \Zset^d_{\geq 0}$ is a (minimal) element of $\safe(s)$ ($\cover(s)$) is $\NP$-hard for 1-player consumption games and $\PSPACE$-hard for 2-player consumption games, even if we restrict ourselves to games where components of all labels are nonzero.
\end{lemma}
\begin{proof}
 We will start by proving that deciding whether $\alpha$ is a (minimal) safe vector in state $s$ is $\PSPACE$-hard for 2-player games, by reduction from QBF (we can use the same reduction for both problems). Once the proof is finished, it will be obvious how to adapt the reduction to prove $\NP$-hardness for 1-player case.

Let $\psi = Q_1 x_1 \dots Q_n x_n \varphi$ be a quantified boolean formula, where $Q_{i} \in \{\forall,\exists\}$ for every $1 \leq i \leq n$, and where $\varphi$ is a quantifier-free formula in conjunctive normal form with clauses $C_1,\dots C_m$ over variables $x_1,\dots,x_n$. We will show how to construct (in time polynomial in size of $\psi$) a 2-player consumption game $G_{\psi}$ of dimension $m$ with a distinguished state $s_1$, such that $\psi$ is true iff the vector $\vec{2n+1}$ is minimal safe vector in $s_1$. 

$G_{\psi}$ is constructed as follows: For every $1 \leq i \leq n$ we have states $s_i,s_{x_i}$ and $s_{\neg x_i}$. We also have states $s_{n+1}$ and $r$. For every $1 \leq i \leq n$ the state $s_{i}$ belongs to Player $\Diamond$ iff $Q_i = \forall$. All other states belong to player $\Box$. Next, for every $1\leq i \leq n$ and for both $L_i \in \{x_i,\neg x_i\}$ we have transitions $s_i \ltran{\alpha_{L_i}} s_{L_i}$
and $s_{L_i} \ltran{\beta_{L_i}} s_{i+1}$. The labels on these transitions are defined in the following way: for every $1 \leq i \leq n$, every $L_i \in \{x_i, \neg x_i\}$ and every $1\leq j \leq m$ we set:
\begin{align*}
 \alpha_{L_i}(j) &= \begin{cases}
                    -2(n-i+1) & \text{if } L_i \text{ is in } C_j\\
		    -1 & \text{otherwise}
                    \end{cases} \\
 \beta_{L_i}(j) &= \begin{cases}
                    \omega & \text{if } L_i \text{ is in } C_j\\
		    -1 & \text{otherwise}
                    \end{cases}
\end{align*}
Finally, we have transitions $s_{n+1} \tran{\vec{-1}} r$ and $r \tran{\vec{\omega}} r$.

Figure \ref{fig:pspace} illustrates the construction for a specific formula $\psi$.  %
 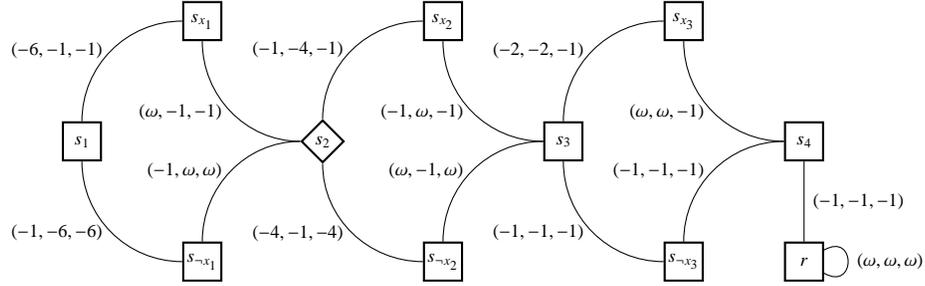
\begin{figure}[t]
\centering
\begin{tikzpicture}[x=1.6cm,y=1.6cm,font=\scriptsize]
\foreach \i/\j/\k in {1/0/sbox,2/2/sdiam,3/4/sbox,4/6/sbox}{
\node(s\i) at (\j,0) [\k] {$s_{\i}$};
}

\node(r) at (6,-1) [sbox] {$r$};

\foreach \i/\j in {1/1,2/3,3/5}{
\node(st\i) at(\j,1) [sbox] {$s_{x_{\i}}$};
}

\foreach \i/\j in {1/1,2/3,3/5}{
\node(sf\i) at(\j,-1) [sbox] {$s_{\neg x_{\i}}$};
}

\foreach \i/\j in {1/{(-6,-1,-1)},2/{(-1,-4,-1)},3/{(-2,-2,-1)}}{
\draw (s\i.north) to [->,tran,out=90,in=180] node[left] {$\j$} (st\i.west);
}

\foreach \i/\j in {1/{(-1,-6,-6)},2/{(-4,-1,-4)},3/{(-1,-1,-1)}}{
\draw (s\i) to [->,tran,out=270,in=180] node[left] {$\j$} (sf\i);
}

\foreach \i/\j/\k in {1/2/{(\omega,-1,-1)},2/3/{(-1,\omega,-1)},3/4/{(\omega,\omega,-1)}}{
\draw (st\i) to [->,tran,out=270,in=180] node[left] {$\k$} (s\j);
}

\foreach \i/\j/\k in {1/2/{(-1,\omega,\omega)},2/3/{(\omega,-1,\omega)},3/4/{(-1,-1,-1)}}{
\draw (sf\i) to [->,tran,out=90,in=180] node[left] {$\k$} (s\j);
}

\draw (s4) to [->,tran] node[right] {$(-1,-1,-1)$} (r);
\draw (r) to [->,tran, loop right] node[right] {$(\omega,\omega,\omega)$} (r);
\end{tikzpicture}
\caption{Game $G_{\psi}$ for formula $\psi = \exists x_1 \forall x_2 \exists x_3 (x_1 \vee \neg x_2 \vee x_3) \wedge (\neg x_1 \vee x_2 \vee x_3) \wedge (\neg x_1 \vee \neg x_2)$.}
\label{fig:pspace}
\end{figure}

Let us note that any strategy of player $\Box$ is safe in configuration $(s,\vec{2n+2})$.

Now, any pair of strategies $\sigma,\pi$ of players $\Box$ and $\Diamond$, respectively, determines a truth assignment $v_{\sigma,\pi} \colon{x_1,\dots,x_n} \rightarrow \{0,1\}$ in obvious way: we set $v_{\sigma,\pi}(x_i) = 1$ iff $\play_{\sigma,\pi}(s_1,\vec{2n+2})$ visits state $s_{x_i}$ (the definition is clearly correct). Thus, formula $\psi$ is true iff player $\Box$ has a strategy $\sigma$ such that for every strategy $\pi$ of player $\Diamond$ the assignment $v_{\sigma,\pi}$ satisfies $\varphi$. But $v_{\sigma,\pi}(\varphi)=1$ iff for every $1 \leq j \leq m$ we have $v_{\sigma,\pi}(L)=1$ for some literal $L$ in $C_j$, which happens if and only if $\play_{\sigma,\pi}(s_1,\vec{2n+2})$ visits state $s_{L}$. Furthermore, $\play_{\sigma,\pi}(s_1,\vec{2n+2})$ visits some $s_{L}$ with $L$ in $C_j$ if and only if the $j$-th resource can be reloaded on $\play_{\sigma,\pi}(s_1,\vec{2n+2})$ \emph{before} state $s_{n+1}$ is reached. Putting these observations together we see that $\psi$ is true if and only if player $\Box$ can \emph{enforce the reload of every resource} before the state $s_{n+1}$ is reached (i.e. during the first $2n$ steps). We claim that this happens if and only if $\vec{2n+1}$ is safe in $s_{1}$. 

Note that no matter what the players do, no resource can decrease by more than $2n$ before the play either reaches $s_{n+1}$ or uses a transition that permits reload of this resource. 
Thus, if $\Box$ can enforce reload of every resource in first $2n$ steps, then any vector $\alpha \geq \vec{2n+1}$ must be safe (whenever player $\Box$ has the opportunity to reload some resource, it suffices to reload it to $2n+1$). On the other hand, if for every strategy $\sigma$ of player $\Box$ the player $\Diamond$ can prevent reload of at least one resource during first $2n$ steps, then, no matter what player $\Box$ does, at least one resource must have value exactly 1 when $s_{n+1}$ is reached (provided that in $s_{1}$ all resources where initialized to $2n+1$). Then player $\Box$ cannot use the transition from $s_{n+1}$ to $r$ and he is thus forced to visit configuration $F$, showing that $\vec{2n+1}$ cannot be safe in $s_1$.

We have proved that $\psi$ is true iff $\vec{2n+1}$ is safe in $s_1$ (and thus we have already proved the $\PSPACE$-hardness of deciding whether a given vector is safe).
Now we show that no vector $\gamma \leq \vec{2n+1}$, $\gamma \neq \vec{2n+1}$ can be safe in $s_{1}$. This will prove that $\psi$ is true iff $\vec{2n+1}$ is \emph{minimal} safe vector in $s_1$. The crucial observation is that no matter what the two players do, every resource decreases by \emph{at least} $2n$ before it can be reloaded \emph{for the first time}. Thus, in order to win, player $\Box$ must start with all resources initialized to at least $2n+1$.

Now to show the $\NP$-hardness of decision problems for 1-player games, it suffices to show a reduction from SAT, i.e. from restricted version of QBF where all quantifiers are existential. It is easy to see that in this restricted case the reduction presented above produces a 1-player game. 

Finally, as we have already observed, if $\vec{2n+1}$ is safe in $s_1$ then player $\Box$ never needs to reload any resource to value greater than $2n+1$ in order to win. Thus, in our reductions $\vec{2n+1}$ is (minimal) element of $\safe(s_1)$ if and only if it is the (minimal) element of $\cover(s_1)$. This gives us the lower bounds on membership problems for $\cover(s)$.\qed
\end{proof}

For better clarity we again restate the important observation about game $G_{\varphi}$ from the proof above. This observation will be useful in the proof of $\DP$-hardness below.
\begin{claim}(1.)
\label{claim:lb-help}
 Let $\varphi$ be a propositional formula (i.e. existentially quantified boolean formula) in CNF with $m$ clauses and $n$ variables. Then the player $\Box$ can enforce the reload of every resource before he reaches the second-to-last state of $G_{\varphi}$ $\Leftrightarrow$ $\varphi$ is satisfiable.
\end{claim}

\begin{lemma}
 The problem, whether a given vector $\alpha$ is a minimal safe vector in a given state $s$ is $\DP$-hard for 1-player games.
\end{lemma}

We will present a reduction from SAT-UNSAT.

Let $(\varphi,\psi)$ be any pair of propositional formulas in CNF, where $\varphi=C_1 \wedge \dots \wedge C_m$ is formula over variables $x_1,\dots,x_n$ and $\psi=C_{m+1} \wedge \dots \wedge C_{m+m'}$ contains only variables $y_{1},\dots,y_{n'}$. Denote $M=m+m'+1$. We will show how to construct (in time polynomial in size of $(\varphi,\psi)$) a 1-player consumption game $G_{\varphi,\psi}$ of dimension $M$ with distinguished state $s_1$ and vector $\xi$, such that $(s_1,\xi)$ is a safe configuration of $G_{\varphi,\psi}$ if and only if $\varphi$ is satisfiable and $\psi$ is unsatisfiable.

Game $G_{\varphi,\psi}$ consists of two gadgets, $H_{\varphi}$ and $H_{\psi}$. Construction of these gadgets is very similar to construction of games $G_{\varphi}$, $G_{\psi}$ from previous lemma (where we treat $\varphi$ and $\psi$ as existentially quantified formulas). However, apart from different dimension of labels there are other subtle differences. Therefore, we give explicit description of these constructions.

Let us start with $H_{\varphi}$: for every $1 \leq i \leq n$ we have states $s_i,s_{x_i}$ and $s_{\neg x_i}$. We also have states $s_{n+1}$ and $r$. For every $1 \leq i \leq n$ there are transitions $s_{i} \tran{} s_{x_i}$, $s_{i} \tran{} s_{\neg x_i}$, $s_{x_i} \tran{} s_{i+1}$ and $s_{\neg x_i } \tran{} s_{i+1}$. Moreover, there is transition $s_{n+1} \tran{} r$. For the time being we leave $r$ without outgoing transition: we will append gadget $H_{\psi}$ to $r$ later.

We now define labeling of transitions in $H_{\varphi}$. For every $1 \leq i \leq n$ and $L_i \in \{x_i,\neg x_i\}$ denote $\alpha_{L_i}$ the label of transition $s_i \tran{} s_{L_i}$ and $\beta_{L_i}$ the label of the single transition outgoing from $s_{L_i}$. Then, for every $1 \leq j \leq M$ we define
\begin{align*}
 \alpha_{L_i}(j) &= \begin{cases}
                    -2(n-i+1)& \text{if } j\leq m \text{ and } L_i \text{ is in } C_j\\
		    -1 & \text{otherwise}
                    \end{cases} \\
 \beta_{L_i}(j) &= \begin{cases}
                    \omega & \text{if }j\leq m \text{ and } L_i \text{ is in } C_j\\
		    -1 & \text{otherwise}
                    \end{cases}
\end{align*}
Finally, we have $s_{n+1} \tran{\vec{-1}} r$.

We will now continue with the construction of $H_{\psi}$. For every $1 \leq i \leq n'$ we have states $t_i,t_{y_i}$ and $t_{\neg y_i}$ together with additional states $t_{n'+1},p,p'$. For every $1 \leq i \leq n'$ and every $L_i\in \{y_i,\neg y_i\}$ we have transitions $t_{i} \ltran{\gamma_{L_i}} t_{L_i} $ and $t_{L_i} \ltran{\delta_{L_i}} t_{i+1}$. Finally we have transitions $t_{n'+1} \tran{\pi} p$, $t_{n'+1} \tran{\pi'} p'$, $p \tran{\vec{\omega}} p$ and $p' \tran{\vec{\omega}} p'$. The labels are defined as follows: we have $\pi = (\underbrace{\omega,\dots,\omega}_{m},\underbrace{-1,\dots,-1}_{m'},-2)$, $\pi' = (\underbrace{\omega,\dots,\omega}_{m},\underbrace{-2,\dots,-2}_{m'},-1)$ and for every $1 \leq i \leq n'$ and every $L_{i}\in \{y_{i},\neg y_{i}\}$ we have:
\begin{align*}
 \gamma_{L_i}(j) &= \begin{cases}
                    -2(n'-i+1) & \text{if } m < j \leq M \text{ and } L_i \text{ is in } C_j\\
		     \omega & \text{if } 1 \leq j \leq m\\
		    -1 & \text{otherwise}
                    \end{cases} \\
 \delta_{L_i}(j) &= \begin{cases}
                    \omega & \text{if } 1 \leq j \leq m\text{ or if }m < j \leq M \text{ and } L_i \text{ is in } C_j\\
		    -1 & \text{otherwise}
                    \end{cases}
\end{align*}

The example of gadgets for specific formulas can be seen on Figure \ref{fig:dp}. Note that $H_{\varphi}$ can be viewed as game $G_{\varphi}$ from proof of Lemma \ref{lem:np}, and thus the observation in Claim (1.) can be applied to this game.

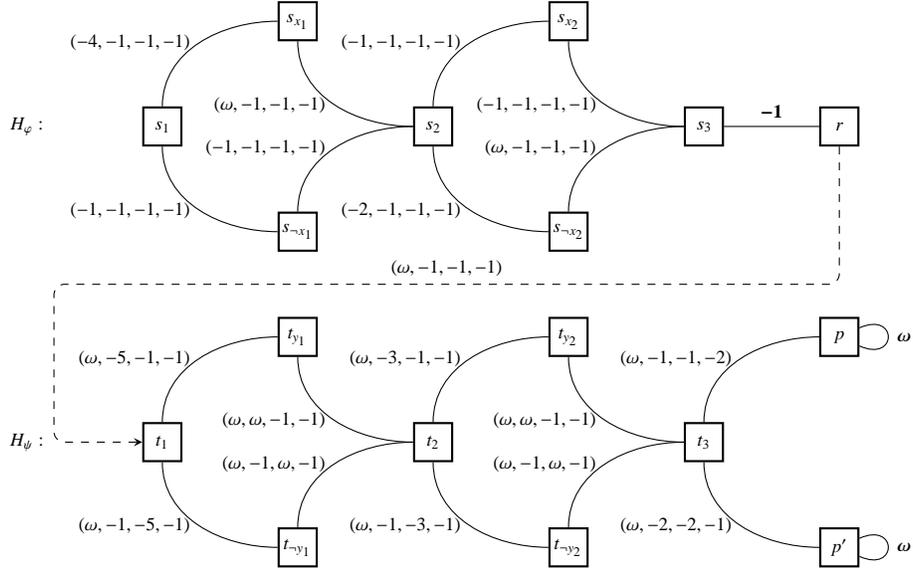
\begin{figure}[t]
\centering
\begin{tikzpicture}[x=1.8cm,y=1.4cm,font=\scriptsize]
 \node(d1) at (0,1) [dummy] {$H_{\psi}:$};
 \node(d2) at (0,4) [dummy] {$H_{\varphi}:$};
 
 \foreach \i/\j/\k/\l in {1/s/1/4,2/s/3/4,3/s/5/4,1/t/1/1,2/t/3/1,3/t/5/1 }{
 \node(\j\i) at (\k,\l) [sbox] {$\j_{\i}$};
 }
 
 \foreach \i/\j/\k/\l/\m in {1/s/2/5/x,2/s/4/5/x,1/t/2/2/y,2/t/4/2/y }{
 \node(t\j\i) at (\k,\l) [sbox] {$\j_{\m_{\i}}$};
 }

 \foreach \i/\j/\k/\l/\m in {1/s/2/3/x,2/s/4/3/x,1/t/2/0/y,2/t/4/0/y }{
 \node(f\j\i) at (\k,\l) [sbox] {$\j_{\neg\m_{\i}}$};
 }

 \node(r) at (6,4) [sbox] {$r$};
 \node(p) at (6,2) [sbox] {$p$};
 \node(pp) at (6,0) [sbox] {$p'$};

 \foreach \i/\j/\k in {1/s/{(-4,-1,-1,-1)},2/s/{(-1,-1,-1,-1)},1/t/{(\omega,-5,-1,-1)},2/t/{(\omega,-3,-1,-1)}}{
 \draw (\j\i) to [tran,out=90,in=180] node[left] {$\k$} (t\j\i);
 }
 \foreach \i/\j/\k in {1/s/{(-1,-1,-1,-1)},2/s/{(-2,-1,-1,-1)},1/t/{(\omega,-1,-5,-1)},2/t/{(\omega,-1,-3,-1)}}{
 \draw (\j\i) to [tran,out=270,in=180] node[left] {$\k$} (f\j\i);
 }
 \foreach \i/\j/\k/\l in {1/s/{(\omega,-1,-1,-1)}/2,2/s/{(-1,-1,-1,-1)}/3,1/t/{(\omega,\omega,-1,-1)}/2,2/t/{(\omega,\omega,-1,-1)}/3}{
 \draw (t\j\i) to [tran,out=270,in=180] node[left] {$\k$} (\j\l);
 }
 \foreach \i/\j/\k/\l in {1/s/{(-1,-1,-1,-1)}/2,2/s/{(\omega,-1,-1,-1)}/3,1/t/{(\omega,-1,\omega,-1)}/2,2/t/{(\omega,-1,\omega,-1)}/3}{
 \draw (f\j\i) to [tran,out=90,in=180] node[left] {$\k$} (\j\l);
 }

 \draw (s3) to [tran] node[above] {$\vec{-1}$} (r);
 \draw (t3) to [tran,out=90,in=180] node[left] {$(\omega,-1,-1,-2)$} (p);
 \draw (t3) to [tran,out=270,in=180] node[left] {$(\omega,-2,-2,-1)$} (pp);
 \draw (p) to [tran, loop right] node[right] {$\vec{\omega}$} (p);
 \draw (pp) to [tran, loop right] node[right] {$\vec{\omega}$} (pp);
 \draw (r) [draw,->,>=stealth,dashed,rounded corners]  -- ++(0,-1.5) -- node[above] {$(\omega,-1,-1,-1)$} ++(-5.8,0) -- ++ (0,-1.5)-- (t1);
\end{tikzpicture}
\caption{Gadgets $H_{\varphi}$ and $H_{\psi}$ for $\varphi = x_1 \vee \neg x_2$ and $\psi = (y_1 \vee y_2) \wedge (\neg y_1 \vee \neg y_2)$. The transition that joins gadgets into game $G_{\varphi,\psi}$ is indicated. The game starts in configuration $(s_1,(5,12,12,13))$.}
\label{fig:dp}
\end{figure}

We now join gadgets $H_{\varphi}$ and $H_{\psi}$ into one game $G_{\varphi,\psi}$ by adding transition $r' \tran{\rho} t_{1}$, where $\rho = (\underbrace{\omega,\dots,\omega,}_{m}-1,\dots,-1)$. Let $\xi$ be such that
\[
 \xi({j})= \begin{cases}
           2n + 1& \text{if } 1 \leq j \leq m \\
	   2n +2n' +4 & \text{if } m< j \leq m'\\
	   2n+2n'+5 & \text{if } j = M
          \end{cases}
\]

We claim that $(\varphi,\psi)$ is a positive instance of SAT-UNSAT iff $\xi$ is minimal safe vector in $s_1$. We will refer to resources with indexes between $m+1$ and $m+m'$ as \emph{intermediate resources}.

Suppose that $\xi$ truly is a minimal safe vector in $s_1$. Then the player must be able to reach one of the states $p$, $p'$ without exhausting any resource. In particular, in order to be able to use transition from $s_{n+1}$ to $r$ the player must be able to reload each of the first $m$ resources \emph{before} he reaches $s_{n+1}$, because each of these resources decrease by at least $2n$ before this state is reached. From Claim (1.) it follows that $\varphi$ is satisfiable. Moreover, when player reaches $t_1$, the value of all intermediate resources is exactly $2n'+2$. Assume, that $\psi$ is satisfiable. Then the player can continue from $t_1$ and reload all intermediate resources before $t_{n'+1}$ is reached, by playing in accordance with the satisfying assignment for $\psi$. The player thus reaches $t_{n'+1}$ with the last resource having value exactly 3 and all other resources having arbitrary large value. He can then use transition from $t_{n'+1}$ to $p'$ decreasing the last resource by 1. It is clear, that the player can use this strategy to reach $p'$ even if he starts with resources initialized to $\xi - (0,\dots,0,1)$, a contradiction with minimality of $\xi$. This shows, that $\psi$ must be unsatisfiable.

Now assume that $(\varphi,\psi)$ is a positive instance of SAT-UNSAT. Since $\varphi$ is satisfiable, the player can reload all of the first $m$ resources before he reaches $s_{n+1}$. Thus, he can reach $t_{1}$ with first $m$ resources having arbitrary large value, the intermediate resources having value exactly $2n'+2$ and the last resource having value exactly $2n'+3$. Then, no matter which path through $H_{\psi}$ the player chooses, he can reach $t_{n'+1}$ without exhausting any resource and with intermediate resources having value at least $2$ and the last resource having value exactly $3$ when $t_{n'+1}$ is reached (he just needs to reload any intermediate resource to $2n'+2$ if he has the possibility to do so along the path). He can then use the transition to $p$ and win the game. This proves that $\xi$ is safe in $s_1$. Suppose that there is some $\gamma \leq \xi$, $\gamma \neq \xi$ that is safe in $s_1$. Clearly, all components of $\gamma$ and $\xi$, except for the last one, have to be equal: each of the first $m$ resources is decreased by at least $2n$ and every intermediate resource is decreased by at least $2n+2n'+3$ before $p$ or $p'$ is reached. Thus, we must have $\gamma({M}) \leq 2n+2n'+4$. But if the player starts with $M$-th resource initialized to $\gamma({M})$, the value of this resource is at most $2$ when $t_{n'+1}$ is reached. Now the only thing the player can possibly do in order to win is to use the transition from $t_{n'+1}$ to $p'$. But to do this, he must reload every intermediate resource before he reaches $t_{n'+1}$. But this is not possible, since $\psi$ is unsatisfiable. Thus, $\gamma$ cannot be safe in $s_1$.

\section{Proofs for Decreasing and One-Player Consumption Games}
\label{app-restricted}

Before we present the proofs for restricted classes of consumption games, let us make the following convenient observation: If we fix the initial configuration, $(s,\alpha)$, then strategies of player $\odot$, where $\odot\in \{\Box,\Diamond\}$, in $\C$ can be seen as functions that for every history $w \in S^{*}S_{\odot}$ return a transition of $\C$ outgoing from the last state of $w$, together with values to which those resources, that can be reloaded by the selected transition, are reloaded. In the following, we will view strategies in this way where it is convenient.

\subsection*{Proof of Lemma~\ref{lem:decreasing-gbuchi}}
From Lemma~\ref{lem-con-Street} we know that $\safe(s)\neq \emptyset$ if 
and only if the player $\Box$ has a winning strategy in $s$ in Streett game 
$\calS_{\C}=(S,E,(S_{\Box},S_{\Diamond}),\A)$ with Streett winning condition 
$\A=\{(G_1,R_1),\ldots,(G_d,R_d)\}$, where
$G_i = \{(s,t) \in E \mid L(s,t)(i) < 0\}$, and
$R_i = $ \mbox{$\{(s,t)~\in~E~\mid L(s,t)(i)~=~\omega\}$} for every 
$1 \leq i \leq d$. Thus, in order to prove Lemma~\ref{lem:decreasing-gbuchi} it suffices to prove that player $\Box$ has a winning strategy in some state $s$ of $\B_{\C}$ iff he has a winning strategy in the corresponding state $s$ in $\calS_{\C}$,

We first observe that a winning strategy in the generalized B\"uchi game 
ensures that every $R_i$ is visited infinitely often, and hence also 
satisfies the Streett condition.
Thus a winning strategy in the generalized B\"uchi game is a winning strategy
in the Streett game.
We now argue for the other direction. 
Consider a winning strategy $\sigma$ for the Streett game. 
We argue that the strategy also ensures the generalized B\"uchi condition.
Consider an arbitrary strategy $\pi$ for player~$\Diamond$, and let 
$w=\play_{\sigma,\pi}(s)$.
Consider $U=\inf(w)$. 
Then $U$ is a scc, and hence must contain cycles. 
Since the game is decreasing, it follows that for every dimension $i$, 
there is at least one transition in $U$ with negative weight or $\omega$ 
in dimension $i$.
Hence for all $1 \leq i \leq d$ either $U \cap R_i\neq \emptyset$ or $U \cap G_i \neq \emptyset$.
Since $\sigma$ is winning for the Streett condition, for all 
$1 \leq i \leq d$ if $U \cap G_i \neq \emptyset$, then $U \cap R_i\neq \emptyset$.
Thus, for all $1 \leq i \leq d$ we have $U \cap R_i \neq \emptyset$. 
Hence all the B\"uchi objectives are satisfied. 
This shows that $\sigma$ is a winning strategy for the generalized B\"uchi game.
The desired result follows.

\subsection*{Proof of Lemma~\ref{lem-claim1}}
Intuitively, it follows from the fact that if a state of $S$ is visited twice (i.e. the play follows a {\it cycle in the state space}) without reloading a ``new'' resource in the cycle (i.e. a resource which has not been already reloaded before the cycle), then player $\Box$ may improve her strategy by simply omitting the cycle. 
More precisely, assume that there is a play which visits the same state in the $i$-th and the $j$-th steps ($i<j$) and
every resource reloaded between the $i$-th and the $j$-the steps has already been reloaded before the $i$-th step.
Then note that the values of the resources that are not reloaded between the $i$-th and the $j$-th steps can only become smaller there, and
as every resource reloaded between the $i$-th and the $j$-th steps has already been reloaded before the $i$-th step, it suffices to make the reloads {\it before} the $i$-th step a bit larger to compensate for removing the reloads {\it between} the $i$-th and the $j$-th step. 

Thus we may modify the strategy $\sigma$ for player $\Box$ so that whenever a play follows the first $i$ steps of the above fixed play, the new strategy, $\sigma'$, starts behaving as if it has already followed $j$ steps of the play (thus removing the part between the $i$-th and the $j$-th steps). Moreover, to be safe, the new strategy has to reload a bit more before reaching the $i$-the step (by e.g.
the total value reloaded between the $i$-th and the $j$-th step of the fixed play).
The resulting strategy will be safe and will reload all resources sooner than the original strategy (at least in some plays). We need the decreasing games in order to be sure that the process of removing cycles will eventually stop (observe that in decreasing games every safe strategy reloads all resources in a bounded number of steps no matter what player $\Diamond$ is doing) and produces a strategy which reloads all resources in at most $d\cdot |S|$ steps. %

Now we present the formal proof.

Let us fix a safe strategy $\sigma$.
Denote by $F_{\sigma,\pi}$ the number of steps the $\play_{\sigma,\pi}(s,\alpha)$ needs to reload all resources.
Denote by $F_{\sigma}$ the maximum $\max_{\pi} F_{\sigma,\pi}$.
It follows from the fact that $\C$ is decreasing that $F_{\sigma}$ is finite for every safe $\sigma$ (otherwise a play won by $\pi$ could be easily constructed due to the fact that once $\sigma$ is fixed, the game is finitely branching).

If $F_{\sigma}\leq d\cdot |S|$, we are done. So assume that $F_{\sigma}>d \cdot |S|$ and that $\sigma$ minimizes $F_{\sigma}$  among all safe strategies. Given $\pi$, denote by $w_{\sigma,\pi}$ the shortest prefix of $\play_{\sigma,\pi}(s,\alpha)$ in which all resources have been reloaded. Denote by $H_{\sigma}$ the number of distinct paths $w_{\sigma,\pi}$ of length $F_{\sigma}$ (i.e., $H_{\sigma}=|\{w_{\sigma,\pi}\mid \pi \in \Pi,\len{w_{\sigma,\pi}}=F_{\sigma}\}|$) and assume that $\sigma$ minimizes $H_{\sigma}$ among all safe strategies.
We show that there is $\sigma'$ such that either $F_{\sigma'}<F_{\sigma}$, or $F_{\sigma'}=F_{\sigma}$ and $H_{\sigma'}<H_{\sigma}$, a contradiction with the minimality of $F_{\sigma}$ and $H_{\sigma}$.

Consider one of the paths $\wh{w}_{\sigma,\pi}=(s_0,\alpha_0)\cdots (s_k,\alpha_k)$ of length $k=F_{\sigma}$. As $F_{\sigma}>d\cdot |S|$, there must be $i<j$ such that $s_i=s_j$ and all resources reloaded between the $i$-th step and the $j$-th step of $\wh{w}_{\sigma,\pi}$ have already been reloaded before the $i$-th step of $\wh{w}_{\sigma,\pi}$. Denote by $B$ the largest value to which any of the resources is reloaded between the $i$-th and $j$-th step of $\wh{w}_{\sigma,\pi}$.
We define choices of a new strategy, $\sigma'$, for all possible histories $w\in S^*S_{\Box}$: 
\begin{itemize}
\item If $w=s_0 \cdots s_{k'}$, where $k'<i$, then $\sigma'$ chooses the same transition as $\sigma$ for the history $w$, but whenever $\sigma$ reloads a resource to $m$, $\sigma'$ reloads the same resource to  $m+B$.
\item If $w=s_0 \cdots s_i t_1\cdots t_{k'}$, then $\sigma'$ chooses the same transition as $\sigma$ for the history $s_0\cdots s_i s_{i+1} \cdots s_j t_1\cdots t_{k'}$. Moreover, whenever $\sigma$ reloads a resource to a value $m$, $\sigma'$ reloads the same resource to the value $m$.
\item For all other histories the strategy $\sigma'$ chooses the same transition as $\sigma$ and reloads resources to the same values as $\sigma$.
\end{itemize}
Note that $\sigma'$ is still safe. Indeed, if 
\[
\play_{\sigma',\pi}(s,\alpha)=(s_0,\alpha_0) \cdots (s_i,\alpha_i)(t_1,\alpha_{i+1})(t_2,\alpha_{i+2})\cdots
\]
then there is a strategy $\pi'$ such that
\[
\play_{\sigma,\pi'}(s,\alpha)=(s_0,\alpha'_0) \cdots (s_i,\alpha'_i)(s_{i+1},\beta_1) \cdots (s_j,\beta_{j-i})(t_1,\alpha'_{i+1})(t_2,\alpha'_{i+2})\cdots
\]
It is easy to observe that $\alpha'_i(j)\leq \alpha_i(j)$ for all $i,j$.

Note also that $F_{\sigma'}\leq F_{\sigma}$. Moreover, 
all plays according to $\sigma'$ that start with the sequence of states $s_0\cdots s_i$ reload all resources in less than $F_{\sigma}$ steps.
Thus either  
$F_{\sigma'}<F_{\sigma}$, or $F_{\sigma'}=F_{\sigma}$ and $H_{\sigma'}<H_{\sigma}$, a contradiction with the minimality of $F_{\sigma}$ and $H_{\sigma}$.

\qed

\subsection*{Proof of Lemma~\ref{lem-claim2}}
The idea is basically the same as in the case of decreasing games. Assume that $\sigma$ does not satisfy the conclusion of the lemma. Then the play $\play_{\sigma}(s,\alpha)$ must follow a cycle in the state space which does not reload any ``new'' resource that has not been reloaded before the cycle. This cycle can be simply omitted and its reloads may be added to the reloads made before the cycle. Eventually, after removing finitely many cycles, all resources that are decreased infinitely many times in $\play_{\sigma}(s,\alpha)$ will be reloaded in the first $d\cdot |S|$ steps, and resources that are decreased finitely many times in $\play_{\sigma}(s,\alpha)$ will be decreased only in the first $d\cdot |S|$ steps.

Formally, let us fix a safe strategy $\sigma$.
Denote by $F_{\sigma}$ the least number of steps the play $\play_{\sigma}(s,\alpha)$ needs to reach a point where every resource is either reloaded, or is never decreased in the future. Clearly, $F_{\sigma}$ must be finite, since otherwise $\sigma$ could not be safe.
If $F_{\sigma}\leq d\cdot |S|$, we are done. So assume that $F_{\sigma}>d\cdot |S|$ and that $\sigma$ minimizes $F_{\sigma}$ among all safe strategies. 

Consider a prefix $w_{\sigma}=(s_0,\alpha_0)\cdots (s_k,\alpha_k)$ of $\play_{\sigma}(s,\alpha)$ of length $k=F_{\sigma}$. As $F_{\sigma}>d\cdot |S|$, there must be $i<j$ such that $s_i=s_j$ and all resources reloaded between the $i$-th step and the $j$-th step of $w_{\sigma}$ have already been reloaded before the $i$-th step of $w_{\sigma}$. Denote by $B$ the largest value to which any of the resources is reloaded between the $i$-th and $j$-th steps of $w_{\sigma}$.
We define choices of a new strategy, $\sigma'$, for all possible histories $w\in S^+$: 
\begin{itemize}
\item If $w=s_0 \cdots s_{k'}$, where $k'<i$, then $\sigma'$ chooses the same transition as $\sigma$ for the history $w$, but whenever $\sigma$ reloads a resource to $m$, $\sigma'$ reloads the same resource to  $m+B$.
\item If $w=s_0 \cdots s_i t_1\cdots t_{k'}$, then $\sigma'$ chooses the same transition as $\sigma$ for the history $s_0\cdots s_i s_{i+1} \cdots s_j t_1\cdots t_{k'}$. Moreover, whenever $\sigma$ reloads a resource to a value $m$, $\sigma'$ reloads the same resource to the value $m$.
\item For all other histories the strategy $\sigma'$ chooses the same transition as $\sigma$ and reloads resources to the same values as $\sigma$.
\end{itemize}
It is easy to see that $\sigma'$ is safe and that $F_{\sigma'}<F_{\sigma}$, a contradiction with the minimality of $F_{\sigma}$.

\qed

\subsection*{Proof of Theorem~\ref{thrm:dec-one-membership}}
The proofs of the lower bounds were given in the previous section, so it remains to prove the upper bounds.

We assume that all states $s$ satisfying $\safe(s)=\emptyset$ have been removed from the game. 
This preprocesing can be done in polynomial time for 1-player and decreasing games (Theorem~\ref{thrm:dec-one-emptiness}).

First, consider decreasing games. According to Lemma~\ref{lem-claim1}, $\alpha\in \safe(s)$ iff there is a safe strategy $\sigma$ of player $\Box$ in configuration $(s,\alpha)$ which reloads all resources in at most $d\cdot |S|$ steps no matter what player $\Diamond$ is doing. However, existence of such a strategy can be decided using polynomial time bounded alternating Turing machine, which implies that the problem belongs to $\PSPACE$. 

Second, consider 1-player consumption games. According to Lemma~\ref{lem-claim2},
$\alpha\in \safe(s)$ iff there is a safe strategy $\sigma$ of player $\Box$ in configuration $(s,\alpha)$ such that $\play_{\sigma}(s,\alpha)$ reaches in at most $d\cdot |S|$ steps a configuration $(t,\beta)$ satisfying the following condition: there is a safe strategy $\sigma'$ in $(t,\beta)$ such that $\play_{\sigma'}(t,\beta)$ does not decrement any recource which has not been reloaded in 
$\play_{\sigma}(s,\alpha)$ before reaching $(t,\beta)$. Note that existence of such $\sigma$ and $\sigma'$ can be decided in non-deterministic polynomial time as follows: first, guess a path $w$ of length at most $d\cdot |S|$ initiated in $(s,\alpha)$ to a configuration $(t,\beta)$ (reloads do not matter too much here, as it suffices to always reload to $2\cdot d\cdot |S|\cdot \ell$). According to what resources have been reloaded in $w$, decide whether $(t,\beta)$ is a safe configuration in a consumption game $\C'$ obtained from $\C$ by pruning all transitions that decrease resources not reloaded in $w$ (this can be done in polynomial time with the algorithm from Theorem~\ref{thrm:dec-one-emptiness}).

Let us now consider the problem whether given $\alpha$ is a minimal element of $\safe(s)$. Clearly, $\alpha$ is a minimal element of $\safe(s)$ iff it is an element of $\safe(s)$ and none of the $d$ vectors $\alpha - \vec{1}_i$, for $1\leq i \leq d$, is safe in $s$. (Recall that $\vec{1}_j$ denotes a vector whose $j$-th component is equal to~$1$ and the other components 
are equal to~$0$.) For decreasing consumption games this amounts to running $d+1$ calls of the polynomial space algorithm for membership, which is again a polynomial space procedure.

To prove $\DP$ upper bound for one-player games, we have to prove that the language of all triples of the form $(\C,s,\alpha)$, where $\C$ is a one-player consumption game, $s$ is a state of $\C$ and $\alpha$ is some minimal element of $\safe(s)$, is an intersection of $\NP$ language $L_1$ and $\coNP$ language $L_2$. For $L_1$ we can take the language of all triples where $\alpha$ is an element of $\safe(s)$, which lies in $\NP$ by the first part of the theorem. For $L_2$ we can take the complement of the language of all triples $(\C,s,\alpha)$ where at least one of the vectors $\alpha - \vec{1}_i$, for $1 \leq i \leq d$, is safe in $s$. Then the complement of $L_2$ lies in $\NP$, because for a given instance $(\C,s,\alpha)$ it suffices to guess $i$ such that $\alpha-\vec{1}_i \in \safe(s)$ and then verify this guess with the non-deterministic polynomial time algorithm for membership presented above.

\subsection*{Proof of Theorem~\ref{thm:min-dist}}

Recall that in the following we have a consumption game $\D=(D,E,(D_{\Box},D_{\Diamond}),L)$ with transitions labeled only by vectors over $\Zset_{\leq 0}$, and a state $r$ of $\D$.

The procedure Min-dist is presented on Figure \ref{fig:reach-alg}.

\begin{figure}[t]%
\begin{procedure}[H]
\DontPrintSemicolon
\caption{Min-dist()($\D$,$s$,$r$)}
\SetKwInOut{Input}{input}
\SetKwInOut{Output}{output}
\SetKwInOut{Input}{input}
\SetKwInOut{Output}{output}
\SetKwData{dist}{dist-set}
\SetKwData{work}{M}
\SetKwData{mindist}{Min-Dist}
\SetKwFunction{minimum}{min-set}
\SetKwData{distiter}{dist}
\SetKwFunction{cwm}{cwm}
\SetKwData{tempset}{temp}

\Input{A one-player consumption game $\D=(D,E,(D_{\Box},D_{\Diamond}),L)$ with labels over $\Zset_{\leq 0}$; states $s,r \in D$}
\Output{The set $\lambda_{\D}(s,r)$}

\lForEach{ $q \in D$}{
  $d[q]$ $\leftarrow \{(\infty,\dots,\infty)\}$, $n[q]\leftarrow d[q]$\;
}
$d[r]$ $\leftarrow \{(1,\dots,1)\}$, $n[r]\leftarrow d[r]$\;
\Repeat{$n[q]=d[q]$ for every $q \in {D}$}{
 \lForEach{$q \in D$}{
 $d[q]\leftarrow n[q]$, $n[q]\leftarrow\emptyset$\;
 }
 \ForEach{$q \in {D}$}{
    let $q\tran{\delta_1} q_1,\dots,q\tran{\delta_m} q_m$ be all transitions in $\D$ outgoing from $q$\;
   \uIf{$q \in D_{\Box}$}{
    \For{ $i = 1$ \KwTo $m$}{\label{alg:reach1}
    $n[q] \leftarrow n[q] \cup (d[q_i]-\delta_i)$\;\label{alg:reach4}
      }
    $n[q]\leftarrow$\minimum{$n[q]$}\label{alg:endfor1}\;
   }
   \Else{ \tempset $\leftarrow d[q_1]$\label{alg:reach5}\;
      \For{$i = 2$ \KwTo $m$}{\label{alg:reach2}
      \tempset$ \leftarrow $\cwm{\tempset,$d[q_i]-\delta_i$} \label{alg:reach3}}
      $n[q]\leftarrow$\minimum{\tempset}\label{alg:endfor2}\;
   }
 }
}
\Return{$d[s]$}
\label{alg:r}
\end{procedure}
\caption{Algorithm for minimal safe distances.}
\label{fig:reach-alg}
\end{figure}

In the following, we denote by $\vec{1}$ the vector $(1,\dots,1)$.

\begin{lemma}
\label{lem:reach-termination}
 Let $a$ be the length of the longest acyclic path in $\D$. Then the procedure Min-dist$(\C,s,r)$ terminates after at most $a$ iterations of the repeat-until loop. Moreover, upon termination we have $d[s]=\lambda_{\D}(s,r)$ for every state $s \in D$. The running time of the algorithm is $\calO\left(|{D}| \cdot b\cdot a \cdot |M|^{2} \right)$ where $b$ is the branching degree of $\D$ and $|M|$ is the maximal possible size of sets $ n[q]$, $d[q]$ and {\tt temp}.
\end{lemma}
\begin{proof}
 For every $j \geq 0$ we say that $\alpha$ is a $j$\emph{-step safe distance} from some state $s$ to some state $r$ if there is a strategy $\sigma$ for player $\Box$ such that for any strategy $\pi$ of player $\Diamond$ the infinite path $\play_{\sigma,\pi}(s,\alpha)$ visits a configuration of the form $(r,\beta)$ \emph{after at most} $j$ steps. We denote $\safe_{\D}^{j}(s,r)$ the set of all $j$-step safe distances from $s$ to $r$ in $\D$. Finally, we denote $\lambda_{\D}^{j}(s,r)$ the set of all minimal elements of $\safe_{\D}^{j}(s,r)$. We start by proving the validity of optimality equations from Section~\ref{sec:reachability}.

\begin{claim}
The following holds:
\begin{enumerate}
 \item For every state $s$ of $\D$ we have 
    \[\lambda_{\D}^{0}(s,r)= \begin{cases}
                           \{\vec{1}\} & \text{if } s=r\\
			   \{(\infty,\dots,\infty)\} &\text{otherwise}
                           \end{cases}
\]
 \item For every $j \geq 1$ and every state $s$ of $\D$ such that $s \ltran{\delta_1} s_1,\dots,s \ltran{\delta_m} s_m$ are all transitions outgoing from $s$ we have
   \[
     \lambda_{\D}^{j}(s,r)= \begin{cases}
                          \text{{\tt min-set}}( \bigcup_{i=1}^{m} (\lambda_{\D}^{j-1}(s_i,r) - \delta_i )) & \text{if } s \in D_{\Box}\\
			  \text{{\tt min-set}}(\text{{\tt cwm}}( \lambda_{\D}^{j-1}(s_1,r) - \delta_1,\dots,\lambda_{\D}^{j-1}(s_m,r) - \delta_m)) & \text{if } s \in D_{\Diamond}
                          \end{cases}
   \]
\end{enumerate}
\end{claim}
(1.) is trivial. In (2.), if $s \in D_{\Box}$, then $\alpha \in \safe_{\D}^{j}(s,r)$ if and only if there is a transition $s \ltran{\delta_i} s_i$ such that $\alpha + \delta_i \in \safe_{\D}^{j-1}(s_i,r)$. Thus, minimal elements of $\safe_{\D}^{j}(s,r)$ are exactly elements of {\tt min-set}$(\bigcup_{i=1}^{m} (\safe_{\D}^{j-1}(s_i,r) - \delta_i )) = \text{\tt min-set}(\bigcup_{i=1}^{m} \text{{\tt min-set}}(\safe_{\D}^{j-1}(s_i,r) - \delta_i )) = \text{\tt min-set}( \bigcup_{i=1}^{m} (\lambda_{\D}^{j-1}(s_i,r) - \delta_i ))$.

Now assume that $s \in D_{\Diamond}$. Denote $R=\text{{\tt min-set}}(\text{{\tt cwm}}( \lambda_{\D}^{j-1}(s_1,r) - \delta_1,\dots,\lambda_{\D}^{j-1}(s_m,r) - \delta_m))$. Note that $\alpha\in\safe_{\D}^{j}(s,r)$ if and only if there are vectors  $\beta_1 \in \safe_{\D}^{j-1}(s_1,r) - \delta_1,\dots,\beta_m \in \safe_{\D}^{j-1}(s_m,r) - \delta_m$ such that $\alpha \geq \beta_i$ for every $1 \leq i \leq m$. That is, $\alpha\in\safe_{\D}^{j}(s,r)$ if and only if the game will be in the safe configuration in the next step, no matter what the player $\Diamond$ does. Obviously, this happens only if there are vectors $\beta_i \in \lambda_{\D}^{j-1}(s_i,r)-\delta_i$, for $1 \leq i \leq m$, such that $\alpha \geq \beta_i$ for every $1 \leq i \leq m$. Thus, $\alpha\in\safe_{\D}^{j}(s,r)$ if and only if there is some $\beta\in R$ such that $\alpha \geq \beta$. Also note that $R\subseteq \safe_{\D}^{j}(s,r)$. Putting these observations together we immediately get $R \subseteq \lambda_{\D}^{j}(s,r)$.

Conversely, if $\alpha \in \lambda_{\D}^{j}(s,r)$, then we already know that there is $\beta \in R$ such that $\alpha \geq \beta$. From the previous paragraph we know that $\beta \in \lambda_{\D}^{j}(s,r)$. Since members of $\lambda_{\D}^{j}(s,r)$ are pairwise incomparable, we get $\alpha = \beta$ and thus $\alpha \in R$. This finishes the proof of the claim.

Let us now consider a computation resulting from the call of Min-dist$(\D,s,r)$. For any $q \in D$ denote $M^j(q)$ the content of $n[q]$ right before the $j$-th iteration of the repeat-until loop. Then it is easy to see, with the help of the previous claim, that $M^{j}(q) = \lambda_{\D}^{j}(q,r)$. (From associativity of operations $\cup$ and {\tt cwm} it follows that computations on lines \ref{alg:reach1}--\ref{alg:endfor1} and \ref{alg:reach5}--\ref{alg:endfor2} correctly compute the sets $\text{{\tt min-set}}( \bigcup_{i=1}^{m} (M^{j-1}(q_i) - \delta_i ))$ and $\text{{\tt min-set}}(\text{{\tt cwm}}( M^{j-1}(q_1) - \delta_1,\dots,M^{j-1}(q_m) - \delta_m))$, respectively.) 

We will now show that for any state $q$ we have $ \lambda_{\D}(q,r) = \lambda_{\D}^{a}(q,r) $, where $a$ is the length of the longest acyclic path in $\D$. This means that the fixed point is reached after at most $a$ iterations of the repeat-until loop and that upon termination we have computed the correct set. %

It suffices to prove that if $\alpha$ is a safe distance from $s$ to $r$, then player $\Box$ has a strategy $\sigma$ such that for any strategy $\pi$ of player $\Diamond$ the path $\play_{\sigma,\pi}(s,\alpha)$ visits some configuration of the form $(r,\beta)$ in at most $a$ steps. The proof closely follows the proof of Lemma \ref{lem-claim1} from section \ref{sec:restricted-bounded-prop}.

Let us fix a safe strategy $\sigma$ such that for any strategy $\pi$ of player $\Diamond$ the path $\play_{\sigma,\pi}(s,\alpha)$ visits a configuration $(r,\beta)$. We will say that such a strategy \emph{ensures the safe visit of $r$}.
Denote by $F_{\sigma,\pi}$ the number of steps the $\play_{\sigma,\pi}(s,\alpha)$ needs to reach configuration of the form $(r,\beta)$.
Denote by $F_{\sigma}$ the maximum $\max_{\pi} F_{\sigma,\pi}$.
It it is easy to see that $F_{\sigma}$ is finite for every $\sigma$ that ensures safe visit of $r$ (otherwise a play won by $\pi$ could be easily constructed due to the fact that without possibility to reload the game is finitely branching).

If $F_{\sigma}\leq a$, we are done. So assume that $F_{\sigma}>a$ and that $\sigma$ minimizes $F_{\sigma}$  among all safe strategies. Given $\pi$, denote by $w_{\sigma,\pi}$ the shortest prefix of $\play_{\sigma,\pi}(s,\alpha)$ on which a configuration of the form $(r,\beta)$ appears. Denote by $H_{\sigma}$ the number of distinct paths $w_{\sigma,\pi}$ of length $F_{\sigma}$ (i.e., $H_{\sigma}=|\{w_{\sigma,\pi}\mid \pi \in \Pi,\len{w_{\sigma,\pi}}=F_{\sigma}\}|$) and assume that $\sigma$ minimizes $H_{\sigma}$ among all safe strategies.
We show that there is $\sigma'$ such that either $F_{\sigma'}<F_{\sigma}$, or $F_{\sigma'}=F_{\sigma}$ and $H_{\sigma'}<H_{\sigma}$, a contradiction with the minimality of $F_{\sigma}$ and $H_{\sigma}$.

Consider one of the paths $w_{\sigma,\pi}=(s_0,\alpha_0)\cdots (s_k,\alpha_k)$ of length $k=F_{\sigma}$. As $F_{\sigma}>a$, there must be $i<j$ such that $s_i=s_j$. %
We define choices of a new strategy, $\sigma'$, for all possible histories (i.e. sequences of states visited before the choice) $w\in D^*D_{\Box}$: 
\begin{itemize}
\item If $w=s_0 \cdots s_{k'}$, where $k'<i$, then $\sigma'$ chooses the same transition as $\sigma$ for the history $w$. %
\item If $w=s_0 \cdots s_i t_1\cdots t_{k'}$, then $\sigma'$ chooses the same transition as $\sigma$ for the history $s_0\cdots s_i s_{i+1} \cdots s_j t_1\cdots t_{k'}$.%
\item For all other histories the strategy $\sigma'$ chooses the same transition as $\sigma$. %
\end{itemize}
The same argument as in proof of Lemma \ref{lem-claim1} reveals that $\sigma'$ still ensures the safe visit of $r$ from $s$, because for every strategy $\pi$ of player $\Diamond$ with
\[
\play_{\sigma',\pi}(s,\alpha)=(s_0,\alpha_0) \cdots (s_i,\alpha_i)(t_1,\alpha_{i+1})(t_2,\alpha_{i+2})\cdots
\]
there is a strategy $\pi'$ such that
\[
\play_{\sigma,\pi'}(s,\alpha)=(s_0,\alpha_0) \cdots (s_i,\alpha'_i)(s_{i+1},\beta_1) \cdots (s_j,\beta_{j-i})(t_1,\alpha'_{i+1})(t_2,\alpha'_{i+2})\cdots
\]
with $\alpha'_i(j)\leq \alpha_i(j)$ for all $i$ and $j$, and $s_k \neq r$ for every $i \leq k \leq j$.

We again have $F_{\sigma'}\leq F_{\sigma}$. Moreover, all plays according to $\sigma'$ that start with the sequence of states $s_0\cdots s_i$ visit a configuration of the form $(r,\beta)$ in less than $F_{\sigma}$ steps. Thus either  
$F_{\sigma'}<F_{\sigma}$, or $F_{\sigma'}=F_{\sigma}$ and $H_{\sigma'}<H_{\sigma}$, a contradiction with the minimality of $F_{\sigma}$ and $H_{\sigma}$.

To finish the proof of Lemma \ref{lem:reach-termination} let us evaluate the complexity of Min-dist. The worst-case performance of the algorithm is primarily affected by computations on lines \ref{alg:reach4} and \ref{alg:reach3}. The worst-case complexities of operations are $\calO(|M|)$ for $M-\alpha$, $\calO(|M|^2)$ for {\tt min-set}$(M)$ and $\calO(|M_1|\cdot|M_2|)$ for $M_1 \cup M_2$ and {\tt cwm}$(M_1,M_2)$. It is then easy to see that the complexity of the algorithm is indeed $\calO\left(|{D}| \cdot b\cdot a \cdot |M|^{2} \right)$.
\qed
\end{proof}

To finish the proof of Theorem \ref{thm:min-dist} we need to prove that number $M$ from Lemma \ref{lem:reach-termination} can be bounded by $(a\cdot\ell)^d$.
It is easy to see that the following holds at any time before the $i$-th iteration of the repeat-until loop (even during the previous iterations): for every $q \in D$ and every $\alpha \in n[q] \cup d[q] \cup $ {\tt temp} we have either $\alpha = (\infty,\dots,\infty)$ or $\alpha \leq (i\cdot \ell,\dots, i\cdot \ell)$. Because the algorithm stops after at most $a$ iterations, the size of sets $n[q]$, $d[q]$ and {\tt temp} is at most $(a\cdot\ell)^d$.

Finally, let us once again restate the property of minimal safe distances that was proved above:
\begin{lemma}
\label{lem:fast-reachability}
 Let $\C$ be a consumption game with transitions labeled by vectors over $\Zset_{\leq}$ and let $s$ and $r$ be states of $\C$.
 If $\alpha$ is a safe distance from $s$ to $r$, then player $\Box$ has a strategy $\sigma$ such that for every strategy $\pi$ of player $\Diamond$ the infinite path $\play_{\sigma,\pi}(s,\alpha)$ visits configuration of the form $(t,\beta)$ in at most $a$ steps, where $a$ is the length of the longest acyclic path in $\C$.
\end{lemma}
We will use the previous observation again in one of the following proofs.

\subsection*{Proof of Theorem~\ref{thm:oneplayer-reach}}
In this section we first give a formal proof of Theorem \ref{thm:oneplayer-reach}. Then we will show how to use this theorem to devise an algorithm for computing minimal elements of $\safe(s)$ in one-player games.

We will often work with more consumption games at
once. To avoid confusion, we denote by $\safe_{\C}(s)$ the set of all
safe vectors in a state $s$ of game $\C$ and by $\ell_C$
the maximal finite $|\delta(i)|$ such that $\delta$ is a label of some
transition in $\C$. We drop the subscript if the game is clear from
the context.  

We start by formal construction of game $\C(\pi)$. First, we introduce some notation.

For infinite path $p = v_{0} \tran{\delta_1} v_{1} \tran{\delta_2} \cdots$ in $\C$ and $1 \leq j \leq d$ we denote $\firstom(p,j)$ the least index $m$ such that $\delta_m(j) = \omega$ (or $\infty$ if $\delta_m(j) \neq\omega$ for every $m$).

For any $I \subseteq \{1,\dots,d \}$ and any vector $\alpha$ we define $\alpha^I$ to be the vector that is obtained from $\alpha$ by zeroing the $i$-th component for every $i \in I$, i.e. $\alpha^{I}(i) = 0$ if $i \in I$ and $\alpha^I(i) = \alpha(i)$ if $i\not\in I$.

For every $I \subseteq \{1,\dots,d\}$ we denote $E_{I}$ the set of all transitions $v \tran{\alpha} v' \in E$ such that $\alpha(i) = 0$ for every $i \in \{1,\dots,d\}\setminus I$. We set $S_I=\{s \in S \mid \exists t \in S\colon (s,t) \in E_{I} \vee (t,s) \in E_I\}$. If $(S_I,E_I)$ is a transition system, then we can define a \emph{zero-subgame induced by} $I$ to be the one-player game $\zerog(\C,I) = (S_I,E_I,(S_I,\emptyset),L_{E_I})$, where $L_{E_I}$ is restriction of $L$ to $E_I$. Note that $\zerog(\C,\{1,\dots,d\})=\C$. We say that a state $s$ of $\C$ is \emph{$I$-safe} if $(S_I,E_I)$ is a transition system and $\safe_{\zerog(\C,I)}(s) \neq \emptyset$. Note that every state is $\{1,\dots,d\}$-safe, since we assume that all states with $\safe(s)=\emptyset$ were removed from game $\C$.

Formally, we have $\C({\pi}) = (S({\pi}),E({\pi}),(S({\pi}),\emptyset),L({\pi}))$, where $S(\pi) = \{$\mbox{$(s,m) \mid s \in S$}$,~0 \leq m \leq d\} \cup \{r,h\}\text{, where }r,h \not\in S$ and where
\begin{align*}
 E({\pi}) &= \{((s,m),(t,m)) \mid (s,t) \in E, L(s,t) \in (\Zset_{\leq 0})^d,~0 \leq m\leq  d \} \\
         &\cup\{((s,m),(t,m+k)) \mid (s,t) \in E,~L(s,t)({\pi_i}) = \omega \text{ for every } m < i \leq m+k\}\\
	 &\cup\{((s,m),r) \mid s \text{ is } \{\pi_1,\pi_2,\dots,\pi_{m}\}\text{-safe} \} \\
	 &\cup\{((s,m),h) \mid s \in S,~0 \leq m \leq d \} \cup\{(r,r),(h,h)\}\\ %
 \end{align*}
The labeling $L({\pi})$ is defined for every $(v,v')\in E({\pi})$ as follows:
\begin{itemize}
 \item if $v = (s,m), v'=(t,m+k)$ for some $0 \leq m < d$ and $0 \leq k < d-m$, then $L({\pi})(v,v') = L(s,t)^{\{\pi_1, \dots, \pi_{m+k}\}}$;
 \item if $v' = r$ then $L({\pi})(v,v')=\vec{0}$;
 \item if $v' = h$ then $L({\pi})(v,v')=\vec{-1}$. 
\end{itemize}
Note that all labels in $\C({\pi})$ are vectors over $\Zset_{\leq 0}^d$ and that branching degree of $\C(\pi)$ is at most $|S|+2$. Also note that the game can be constructed in time polynomial in size of $\C$ using the polynomial time algorithm for checking emptiness (Theorem \ref{thrm:dec-one-emptiness}).

We now prove the second part of the theorem.
It clearly suffices to prove that $\safe_{\C}(s,\pi)=\safe_{\C(\pi)}((s,0),r)$.

Suppose that $\alpha\in\safe_{\C}(s,\pi)$. Fix a safe strategy $\sigma$ in $(s,\alpha)$ such that $\play_{\sigma}(s,\alpha)$ matches $\pi$. Denote $\play_{\sigma}(s,\alpha)=(s,\alpha) \tran{\delta_1} (s_1,\alpha_1) \tran{\delta_2} (s_2,\alpha_2) \cdots$ and denote $w=s \tran{\delta_1} s_1 \tran{\delta_2} s_2 \cdots$ the corresponding path in the state space of $\C$.
Finally, let $I = $\mbox{$\{i \in \{1,\dots,d\}\mid \delta_k(i) = \omega \text{ for some } k\geq 1\}$}.

Since $\sigma$ is safe, there must be number $K\geq \firstom(w,\pi_{|I|})$ such that for every $j \geq K$ we have $\delta_j(i) = 0$ for every $i \in \{1,\dots,d\}\setminus I$.
But this means that the state $s_K$ is $I$-safe. Thus, the following sequence $w'$ must be a path in the state space of $\C(\pi)$:
\begin{align*}
 w' &= (s,0)\tran{\delta_1} \cdots (s_{\firstom(p,\pi_1)-1},0) \tran{\delta_{\firstom(p,\pi_1)}^{\pi_1}} (s_{\firstom(p,\pi_1)},n_{\firstom(p,\pi_1)}) \tran{}\cdots\\
    &\tran{} (s_{\firstom(p,\pi_m)-1},n_{\firstom(p,\pi_m)-1}) \tran{\delta^{\pi_1,\dots,\pi_{m}}_{\firstom(p,\pi_m)}}(s_{\firstom(p,\pi_m)},n_{\firstom(p,\pi_m)}) \tran{} \cdots\\
    &\tran{} (s_{\firstom(p,\pi_{|I|})-1},n_{\firstom(p,\pi_{|I|})-1}) \tran{\vec{0}} r \tran{\vec{0}} r \tran{\vec{0}} r\cdots,
\end{align*}
where $n_{l}$ denotes number of distinct resources that were reloaded at least once along the path $w$ during the first $l$ steps.
Moreover, for every $i \in I$ the decrease of the $i$-th resource before the first reload in $\C(\pi)$ along the path $w'$ is the same as decrease of $i$-th resource before the first reload in $C$ along the path $w$. For $i \not \in I$ the decrease of the $i$-th resource in $C(\pi)$ along the path $w'$ is the same as  decrease of this resource along the path $w$ in $\C$. Now if we define strategy $\sigma'$ in $\C(\pi)$ to simply follow the path $w'$ until $r$ is reached (for histories that are not prefixes of $w'$, the strategy $\sigma'$ is defined arbitrarily), then it is easy to see that starting with resources initialized to $\alpha$ no resource is ever depleted, and thus $\play_{\sigma'}((s,0),\alpha)$ visits configuration of the form $(r,\beta)$.
Thus, $\alpha$ is a safe distance from $(s,0)$ to $r$.

On the other hand, assume that $\alpha$ is a safe distance from $(s,0)$ to $r$ in $\C(\pi)$. Let $\sigma$ be a strategy such that $\play_{\sigma}((s,0),\alpha)$ reaches configuration of the form $(r,\beta)$. Let $ \play_{\sigma}((s,0),\alpha)= ((s,0),\alpha) \tran{\delta_1} ((s_1,i_1),\alpha_1) \cdots ((s_k,i_k),\alpha_k) \tran{\vec{0}} (r,\alpha_k)^{\omega}$ and let $w=(s,0)\tran{\delta_1} (s_1,i_1) \tran{\delta_2} (s_2,i_2) \cdots (s_k,i_k) \tran{\vec{0}} r$ be the corresponding finite path in state space of $\C(\pi)$. Then there is a finite path $w'= s \tran{\delta_1 '} s_1 \cdots s_{k-1} \tran{\delta_{k}'} s_{k}$ in the state space of $\C$. If we denote $I$ the set of resources that were reloaded somewhere along the path $w'$, then for every $i \in I$ the decrease of $i$-th resource before the first reload in $\C$ along the path $w'$ is the same as decrease of $i$-th resource before the first reload in $\C(\pi)$ along the path $w$. For $i \not \in I$ we have that decrease of the $i$-th resource along the path $w'$ in $\C$ is the same as decrease of this resource along the path $w$ in $\C(\pi)$. Moreover, we know that state $s_{k}$ is $I$-safe. Thus, the player has the following safe strategy $\sigma'$ in $(s,\alpha)$: play along the path $w'$ until its end is reached. If it is possible to reload any resource along this path, always reload it to value $k\cdot M$, where $M=\ell \cdot |S|\cdot d$. When the play reaches the state $s_{k}$, switch to a safe strategy for $(s_k,(M,\dots,M))$ in game $\zerog(\C,I)$. Such a strategy must exist due to the Corollary \ref{cor-frp}.

Strategy $\sigma'$ is safe because when $s_{k} $ is reached, resources in $I$ have value at least $M$, and resources not included in $I$ are never decreased after the $k$-th step. Furthermore, all first reloads happen during the first $k$ steps. Thus, the order of first reloads is the same as in $w$. This proves that $\alpha \in \safe_{\C}(s,\pi)$.

\begin{figure}[t]%
\begin{procedure}[H]
\caption{Min-safe()($\C$,$s$)}
\DontPrintSemicolon
\SetKwInOut{Input}{input}
\SetKwInOut{Output}{output}
\SetKwData{dist}{dist-set}
\SetKwData{work}{M}
\SetKwFunction{mindist}{Min-dist}
\SetKwFunction{minimum}{min-set}
\Input{A $d$-dimensional one-player consumption game $\C$ such that $\safe(t) \neq \emptyset$ for every state of $\C$, state $s$ of $\C$ }
\Output{A set of all minimal elements of $\safe(s)$}

\work$\leftarrow\emptyset$, \dist$\leftarrow \emptyset$\;
\ForEach{$\pi \in \Pi(d)$}{
      construct a game $\C({\pi})$ with a distinguished state $r$\;
      \dist$\leftarrow$\mindist{$\C({\pi})$,$(s,0)$,$r$}\;
      \work $\leftarrow$ \minimum{\dist $\cup$ \work}\;
}
\Return{\work}
\end{procedure}
\caption{Computing minimal elements of $\safe(s)$}
\label{fig:alg-oneplayer}
\end{figure}

Now, we will show, how to use Theorem \ref{thm:oneplayer-reach} to construct procedure Min-safe that computes the set of minimal elements of $\safe(s)$ for a state $s$ of a given one-player consumption game $\C$.
The procedure Min-safe is presented on Figure \ref{fig:alg-oneplayer}. Its correctness is a straightforward consequence of Theorem \ref{thm:oneplayer-reach}.

Let us now discuss the complexity of procedure Min-safe. For a fixed permutation $\pi$ the state space of $\C({\pi})$ can be constructed in $\calO(|S|\cdot d)$ time. In order to compute transitions in $\C({\pi})$ we need to determine whether $\safe_{\zerog(\C,\{\pi_1,\dots,\pi_j\})}(s) \neq \emptyset$, for every $s \in S$ and every $1 \leq j \leq d$. This amounts to solving $|S|\cdot d$ instances of nonemptiness problem for Streett automata, for automata with at most $|S|$ states and $|E|$ transitions, so the complexity of this step is $\calO((|S|^2 \cdot d^2 +|E|\cdot|S|\cdot d)\cdot\min\{|S|,d\})$.

The branching degree of $\C({\pi})$ is at most $|S|+2$ and no acyclic path in $\C({\pi})$ is longer than $|S|\cdot d + 1$. Also, $\ell_{\C(\pi)}\leq \ell_{\C}$. Thus, single call of Min-dist$(\C({\pi}),(s,0),r)$ has complexity \mbox{$\calO(|S|^{2d+3}\cdot \ell^{2d}\cdot d^{2d+2})$}. Moreover, from Theorem \ref{thm:min-dist} we get that the size of the set {\tt M} never exceeds $|S|^d\cdot d^d\cdot \ell^d$, so a single execution of line 5 has complexity at most $\calO(|S|^{2d} \cdot d^{2d} \cdot \ell^{2d})$. We can see that the overall worst-case running time of the algorithm is \mbox{$\calO(d!\cdot |S|^{2d+3}\cdot \ell^{2d}\cdot d^{2d+2})$}.

\subsection*{Proof of Theorem~\ref{thm:decreasing}}
We will again start with the proof of Theorem \ref{thm:decreasing} and then we will discuss the complexity of the resulting algorithm.

First, let us formally define the game $\wh{C}$.
Recall that for any $I \subseteq \{1,\dots,d \}$ and any vector $\alpha$ we denote $\alpha^I$ the vector $\beta$ that is obtained from $\alpha$ by zeroing the $i$-th component for every $i \in I$, i.e. $\alpha^{I}(i) = 0$ if $i \in I$ and $\alpha^I(i) = \alpha(i)$ if $i\not\in I$.

We put $\wh{\C} = (\wh{S},\wh{E},(\wh{S}_{\Box},\wh{S}_{\Diamond}),\wh{L})$, 
where
\begin{align*} 
 \wh{S} &= \{(s,I) \mid s \in S, I \subseteq\{1,\dots,d\} \}\cup \{r\}\text{, where }r \not\in S\\
 \wh{E} &= \{((s,I),(t,I)) \mid (s,t) \in E,~L(s,t) \in (\Zset_{\leq 0})^d\} \\
         &\cup\{((s,I),(t,J)) \mid (s,t) \in E,~J=I \cup\{j\mid L(s,t)({j}) = \omega\}\}\\
	 &\cup\{((s,\{1,\dots,d\}),r) \mid s \in S \}\\
	 &\cup\{(r,r)\}\\ %
 \end{align*}
and where the labeling $\wh{L}$ is defined for every $(v,v')\in \wh{E}$ as follows:
\begin{itemize}
 \item if $v = (s,I), v'=(t,J)$ for some $I \subseteq J \subseteq \{1,\dots,d\}$ 
 then $\wh{L}(v,v') = \alpha^{J}$;
 \item if $v' = r$ then $\wh{L}(v,v')=\vec{0}$.
\end{itemize}
Finally, the partition of state space $(\wh{S}_{\Box},\wh{S}_{\Diamond})$ is defined as follows: a state of the form $(s,I)$ belongs to player $\Box$ iff $s \in S_{\Box}$ or $I=\{1,\dots,d\}$.

Note that the game $\wh{\C}$ can be constructed in time $\calO(2^d \cdot(|S|+|E|))$, the branching degree of $\wh{\C}$ is at most $|S|+1$ and the length of the longest acyclic path in $\wh{\C}$ is at most $|S|\cdot d +1$.

To prove the first part of the theorem, it suffices to prove that $\safe_{\C}(s) = \safe_{\wh{\C}}((s,\emptyset),r)$. The equality of sets of minimal elements then follows.

Suppose that $\alpha\in\safe_{\C}(s)$.
Let $\sigma$ be some safe strategy in $(s,\alpha)$ in game $\C$. We define a strategy $\wh{\sigma}$ for player $\Box$ in $\wh{\C}$. Note that since we have no reloads in $\wh{\C}$, for the fixed initial configuration $((s,\emptyset),\alpha)$ all strategies in $\wh{\C}$ can be viewed as functions from finite sequences of states to transitions in $\wh{\C}$ (current configuration of the game can be always inferred from knowledge of this sequence and initial configuration). Therefore, we assume that initial configuration $((s_0,\emptyset),\alpha_0)$ with $s_0 = s$ and $\alpha_0 = \alpha$ was fixed. We will define the choice of strategy $\wh{\sigma}$ for every history $w$ in $\wh{\C}$.

For every finite path $w=(s_0,\emptyset)(s_1,I_1)(s_2,I_2)\cdots(s_m,I_m)$ in the state space of $\wh{\C}$ we denote $w'$ the corresponding path $s_0 s_1 \dots s_m$ in the state space of $\C$. If $I_m = \{1,\dots,d\}$, then we require $\wh{\sigma}$ to select the transition leading to $r$. Otherwise, let $(s_m,s')$ be the transition in $\C$ selected by $\sigma$ given a history $w$. There is a unique transition of the form $((s_m,I_m),(s',I'))$ in $\wh{\C}$ and we require $\wh{\sigma}$ to select exactly this transition given a history $w$.

We claim that that for any strategy $\pi$ of player $\Diamond$ in $\wh{\C}$ the path $\play_{\wh{\sigma},\pi}((s_0,\emptyset),\alpha)$ reaches configuration a configuration $(r,\beta)$ for some $\beta$ (and thus $\alpha\in\safe_{\wh{\C}}((s,\emptyset),r)$). Fix a strategy $\pi$ of player $\Diamond$. Then there is a strategy $\pi'$ of player $\Diamond$ in $\C$ such that the decrease of the $i$-th resource before its first reload on $\play_{\sigma,\pi'}(s,\alpha)$ is the same as the decrease of the $i$-th resource on $\play_{\wh{\sigma},\pi}((s,\emptyset),\alpha)$, for every $1 \leq i \leq d$. Thus, the $\play_{\wh{\sigma},\pi}((s,\emptyset),\alpha)$ cannot reach configuration $F$, because otherwise $\play_{\sigma,\pi'}(s,\alpha)$ would also visit $F$, a contradiction with $\sigma$ being safe in $(s,\alpha)$. The only other possibility for the $\play_{\wh{\sigma},\pi}((s,\emptyset),\alpha)$ not to reach configuration of the form $(r,\beta)$ is that $\play_{\wh{\sigma},\pi}((s,\emptyset),\alpha)$ visits infinitely many states of the form $(t,I)$, where $I \subset \{1,\dots,d\}$. But then there is a strategy $\pi''$ of player $\Diamond$ in $\C$ such that some resource is never reloaded along the $\play_{\sigma,\pi''}(s,\alpha)$. Because $\C$ is decreasing, this would again imply that $\play_{\sigma,\pi'}(s,\alpha)$ visits $F$, a contradiction with $\sigma$ being safe in $(s,\alpha)$.

Now suppose that $\alpha \in \safe_{\wh{\C}}((s,\emptyset),r)$. Fix a strategy $\sigma$ in $((s,\emptyset),\alpha)$ such that $\play_{\sigma,\pi}((s_0,\emptyset),\alpha)$ reaches $(r,\beta)$ (for some vector $\beta$) for any strategy $\pi$ of $\Diamond$ in at most \mbox{$|S|\cdot d$} steps. Such a strategy exists due to the Lemma \ref{lem:fast-reachability}. We define a safe strategy $\sigma'$ in $(s,\alpha)$ in $\C$. For every finite path $w = s_0s_1\cdots s_m$ in the state space of $\C$, such that not every resource is reloaded along $w$, we denote $w'$ the unique finite path in the state space of $\wh{\C}$, such that $w' = (s_0,\emptyset)(s_1,I_1)\cdots (s_m,I_m)$ for suitable $I_j$, where $1 \leq j \leq m$. %
Now we require $\sigma'(w)$ to select the unique transition of the form $s_{m} \tran{} q$ where $q$ is such that  $\sigma(w') = (q,J)$ for some $J$. If this transition permits reload of any resource, the strategy reloads it to $|S|\cdot d \cdot M$, where $M = \ell \cdot |S| \cdot d$. Once all resources were reloaded at least once, $\sigma'$ starts to behave like some safe strategy in configuration $(q,(M,\dots,M))$, where $q$ is the current state of the game (such a strategy must exist  due to the Corollary \ref{cor-frp}). 

We claim that strategy $\sigma'$ is safe in $(s,\alpha)$. Since $\sigma$ ensured that the play always reaches some configuration of the form $(r,\beta)$ (and thus also some configuration of the form $((q,\{1,\dots,d\}),\beta')$), the strategy $\sigma'$ will always start to behave as some safe strategy in $(q,(M,\dots,M))$, for some state $q$, no matter what the player $\pi$ does. When this happens, all resources have value at least $M$, because the switch happens after at most $|S|\cdot d$ steps. If there is some strategy $\pi$ of player $\Diamond$ such that $\play_{\sigma',\pi}(s,\alpha)$ visits configuration $F$ before all resource could be reloaded, then there is also strategy $\pi'$ of player $\Diamond$ in $\wh{\C}$ such that $\play_{\sigma,\pi'}((s,\emptyset),\alpha)$ visits $F$, a contradiction with $\sigma$ being a safe strategy. It follows that $\sigma'$ is safe in $(s,\alpha)$.

\paragraph*{Complexity}
Let us now discuss the complexity of the algorithm. The  construction of $\wh{\C}$ takes time
$\calO(2^d \cdot(|S|+|E|))$. 
The game $\wh{\C}$ has $\calO(2^d \cdot|S|)$ states and branching degree at most $|S|+2$. 
We also have $\ell_{\wh{\C}} \leq \ell_{\C}$ (i.e., the maximal absolute
value of weight in $\wh{C}$ is bounded by the maximal absolute weight 
in $\C$). 
Finally, the length of every acyclic path in $\wh{\C}$ is bounded by 
$|S|\cdot d + 1$, because from every state $(t,I)$ we can reach only states 
$r,h$ and states of the form $(t,J)$ where $J=I$ or $|J|>|I|$. Thus, 
by Theorem~\ref{thm:min-dist} the call Min-dist$(\wh{\C},(s,\emptyset),r)$ takes time 
$\calO(2^d\cdot d^{2d+2}\cdot|S|^{2d+3}\cdot \ell^{2d})$. 
This is also the total worst case complexity of our algorithm.

\end{document}